\newif\ifExtended
\Extendedtrue

\documentclass[sigconf]{acmart}

\ifExtended
\setcopyright{none}
\acmConference[]{}{}{} 
\renewcommand\footnotetextcopyrightpermission[1]{} 
\else
\copyrightyear{2023}
\acmYear{2023}
\setcopyright{acmlicensed}\acmConference[CCS '23]{Proceedings of the 2023
ACM SIGSAC Conference on Computer and Communications
Security}{November 26--30, 2023}{Copenhagen, Denmark}
\acmBooktitle{Proceedings of the 2023 ACM SIGSAC Conference on Computer
and Communications Security (CCS '23), November 26--30, 2023, Copenhagen,
Denmark}
\acmPrice{15.00}
\acmDOI{10.1145/3576915.3623141}
\acmISBN{979-8-4007-0050-7/23/11}
\fi

\usepackage{color}
\usepackage{algorithm}
\usepackage{listings}
\usepackage{algpseudocode}

\usepackage[scaled=0.85]{beramono}

\definecolor{lcol}{rgb}{0,0,0.5}

\newif\ifFinal
\Finaltrue  

\newcommand{\myparagraph}[1]{\paragraph{#1}}

\usepackage{pifont}
\usepackage{cite}
\usepackage{color}
\usepackage{array}
\usepackage{mathpartir}
\usepackage{xifthen}

\usepackage{url}
\usepackage{xspace}
\usepackage{amsmath}
\usepackage{multicol}
\usepackage{nameref}
\usepackage{stmaryrd}

\newcommand{\SecC}{\textsc{SecC}\xspace}
\newcommand{\SecCSL}{\textsc{SecCSL}\xspace}
\newcommand{\Tool}{\textsc{Verdeca}\xspace}
\newcommand{\Veronica}{\textsc{Veronica}\xspace}

\usepackage{tikz}
\usetikzlibrary{patterns,arrows,decorations.pathreplacing}
\usepackage{subcaption}

\usepackage{color}

\newcommand{\FIXME}[1]{\colorbox{yellow}{\textbf{\emph{FIXME:}}} \textbf{\emph{#1}}}

\newcommand{\addition}[1]{{#1}}

\usepackage{enumitem}
\setitemize{noitemsep,nosep,leftmargin=3mm}

\usepackage[capitalize]{cleveref}

\definecolor{dscolor}{rgb}{0.1,0.1,0.9}

\newcommand{\Triple}[3]{\{#1\}\ #2\ \{#3\}}
\newcommand{\prove}[4]{\vdash_{#1} \Triple{#2}{#3}{#4}}
\newcommand{\audit}[5]{\vdash_{#1} \Triple{#2}{#3}{#4} \triangleright #5}

\newcommand{\List}[1]{\langle#1\rangle} 
\newcommand{\Concat}[2]{#1 \cdot #2} 
\newcommand{\Nil}{\List{}}

\newcommand{\Run}[4]{(\mathbf{run}\ #1, #2, #3, #4)}
\newcommand{\Stop}[3]{(\mathbf{stop}\ #1, #2, #3)}
\newcommand{\Abort}{(\mathbf{abort})}
\newcommand{\Step}[3]{#1 \stackrel{#2}{\longrightarrow} #3}
\newcommand{\tranStar}[1]{\mathrel{\overset{{#1}}{\protect{\raisebox{0pt}[.6ex][0pt]{$\longrightarrow$}}}\raisebox{.6ex}[0pt][0pt]{\small$*$}}}
\newcommand{\Steps}[3]{#1 \tranStar{#2} #3}
\newcommand{\Out}[2]{\mathbf{Out}\ #1\ #2}
\newcommand{\Trace}[1]{\mathbf{Trace}\ #1}
\newcommand{\Fst}{\mathbf{L}}
\newcommand{\Snd}{\mathbf{R}}
\newcommand{\Tick}{\tau}
\newcommand{\Load}[1]{\mathbf{Load}\ #1}
\newcommand{\Store}[1]{\mathbf{Store}\ #1}
\newcommand{\Assumption}[1]{\mathbf{Assm}\ #1}

\newcommand{\Emp}{\mathbf{emp}}
\newcommand{\pto}[2]{#1\ \mapsto\ #2}

\newcommand{\Exists}[2]{\exists #1.\ #2}

\newcommand{\haslabel}{::}

\definecolor{seccol}{rgb}{0.9,0.9,0.9}
\newcommand{\seccheck}[1]{\colorbox{seccol}{$#1$}}

\newcommand{\defeq}{\hat{=}}
\newcommand{\union}{\cup}

\newcommand{\sem}[2]{\llbracket #1 \rrbracket_{#2}}
\newcommand{\emptyheap}{\varnothing}

\newcommand{\tr}{\mathit{tr}} 
\newcommand{\ep}{{e_p}}
\newcommand{\ev}{{e_v}}
\newcommand{\el}{{{\ell'}}} 

\newcommand{\ITE}[3]{\textbf{if}\ #1\ \textbf{then}\ #2\ \textbf{else}\ #3}
\newcommand{\WHILE}[2]{\textbf{while}\ #1\ \textbf{do}\ #2}
\newcommand{\ASSUME}[1]{\textbf{assume}\ #1}
\newcommand{\ASSERT}[1]{\textbf{assert}\ #1}

\newcommand{\OUTPUT}[2]{\textbf{output}\ #1\ #2}
\newcommand{\LOCK}[1]{\textbf{lock}\ #1}
\newcommand{\UNLOCK}[1]{\textbf{unlock}\ #1}
\newcommand{\LOAD}[2]{#1 := [#2]}
\newcommand{\STORE}[2]{[#1] := #2}
\newcommand{\ASSIGN}[2]{#1 := #2}
\newcommand{\SKIP}{\textbf{skip}}
\newcommand{\TRACE}[1]{\textbf{trace}\ #1}

\newcommand{\D}{\mathcal{D}} 
\newcommand{\Dpre}{\varphi_\D}
\newcommand{\Dpost}{\rho_\D}
\newcommand{\policy}[2]{#1 \leadsto #2}
\newcommand{\History}{\mathcal{H}} 
\newcommand{\Reserve}{\mathit{Reserve}}

\newcommand{\st}{\mathit{st}}
\newcommand{\met}{\mathit{met}}
\newcommand{\resmet}{\mathsf{resmet}}

\newcommand{\visible}[1]{\mathsf{visible}_\ell(#1)}

\newcommand{\countbudget}[1]{\mathsf{count\_budget}(#1)} 
\newcommand{\length}{\mathsf{length}} 
\newcommand{\sumfunc}{\mathsf{sum}} 
\newcommand{\fv}[1]{\mathsf{fv}(#1)} 
\newcommand{\inv}[1]{\mathsf{inv}(#1)} 
\newcommand{\invs}[1]{\mathsf{invs}(#1)} 
\newcommand{\dom}[1]{\mathsf{dom}(#1)} 
\newcommand{\uncertainty}[7]{\mathsf{uncertainty}_{#1}(#2,#3,#4,#5,#6,#7)}

\newcommand{\assumedrelease}[7]{\mathsf{assumed\text{-}release}_{#1}(#2,#3,#4,#5,#6,#7)}
\newcommand{\policyrelease}[8]{\mathsf{policy\text{-}release}_{#1}(#2,#3,#4,#5,#6,#7,#8)}

\newcommand{\secure}[3]{\mathsf{secure}_{#1}^{#2}(#3)}
\newcommand{\schedcteq}[3]{#2 \approx_{#1} #3}

\newcommand{\cut}[1]{_{|#1}}

\newcommand{\high}{\textbf{\texttt high}\xspace}
\newcommand{\low}{\textbf{\texttt low}\xspace}
\newcommand{\requires}{\textbf{\texttt requires}\xspace}
\newcommand{\ensures}{\textbf{\texttt ensures}\xspace}

\newcommand{\trueconst}{\textbf{\texttt true}\xspace}
\newcommand{\falseconst}{\textbf{\texttt false}\xspace}

\newcommand{\result}{\textbf{\texttt result}\xspace}

\newcommand{\ann}[1]{\texttt{\_(#1)}\xspace}

\newcommand{\avggetinput}{\texttt{avg\_get\_input()}\xspace}
\newcommand{\printaverage}{\texttt{print\_average()}\xspace}
\newcommand{\avglock}{\texttt{avg\_lock()}\xspace}
\newcommand{\avgunlock}{\texttt{avg\_unlock()}\xspace}
\newcommand{\avgsumthread}{\texttt{avg\_sum\_thread()}\xspace}
\newcommand{\avgdeclassthread}{\texttt{avg\_declass\_thread()}\xspace}
\newcommand{\getrealloc}{\texttt{get\_real\_loc()}\xspace}
\newcommand{\printloc}{\texttt{print\_loc()}\xspace}
\newcommand{\varp}{\texttt{p}\xspace}
\newcommand{\vars}{\texttt{s}\xspace}
\newcommand{\budget}{\texttt{budget}\xspace}
\newcommand{\addnoise}{\texttt{add\_noise()}\xspace}
\newcommand{\logreplenish}{\texttt{log\_replenish()}\xspace}

\newcommand{\noiselat}{\mathit{nlat}}
\newcommand{\noiselon}{\mathit{nlon}}

\newcommand{\Consumed}[2]{\mathbf{Consumed}\,(#1,#2)}
\newcommand{\Replenished}{\mathbf{Replenished}}

\newcommand{\vard}{\texttt{d}\xspace}
\newcommand{\logbid}{\texttt{log\_bid()}\xspace}
\newcommand{\logclosed}{\texttt{log\_closed()}\xspace}
\newcommand{\printresult}{\texttt{print\_result()}\xspace}

\newcommand{\bidid}{\mathit{id}}
\newcommand{\bidqt}{\mathit{qt}}

\newcommand{\ismax}{\mathsf{ismax}}
\newcommand{\contains}{\mathsf{contains}}

\newcommand{\Running}[2]{\mathbf{Run}\,(#1,#2)}
\newcommand{\Finished}{\mathbf{Fin}}

\newif\ifdelta
\deltafalse

\ifdelta
\newcommand\changed[2]{\textcolor{red}{#1}\textcolor{blue}{#2}}
\newenvironment{added}{\par\color{blue}}{\par}
\newenvironment{deleted}{\par\color{red}}{\par}
\else
\newcommand\changed[2]{#2}
\newenvironment{added}{\relax}{\relax}
\newenvironment{deleted}{\relax}{\relax}
\usepackage{version}
\excludeversion{deleted}
\fi

\begin{document}
%


\title{Assume but Verify: Deductive Verification of Leaked Information in Concurrent Applications\ifExtended\ (Extended Version)\fi}

%
%

\ifFinal
\author{Toby Murray}
\email{toby.murray@unimelb.edu.au}
\orcid{0000-0002-8271-0289}
\affiliation{%
  \institution{University of Melbourne}
  \country{Australia}
}

\author{Mukesh Tiwari}
\authornote{This work was carried out while the author was at the University of Melbourne.}
\email{mt883@cam.ac.uk}
\orcid{0000-0001-5373-9659}
\affiliation{%
  \institution{University of Cambridge}
  \country{United Kingdom}
}

\author{Gidon Ernst}
\email{gidon.ernst@lmu.de}
\orcid{0000-0002-3289-5764}
\affiliation{%
  \institution{LMU Munich}
  \country{Germany}
}

\author{David A. Naumann}
\email{naumann@cs.stevens.edu}
\orcid{0000-0002-7634-6150}
\affiliation{%
  \institution{Stevens Institute of Technology}
  \country{USA}
}

\renewcommand{\shortauthors}{Murray et al.}
\fi 
\begin{abstract}
  We consider the problem of specifying and proving the security of non-trivial, concurrent programs that
  intentionally leak information.
We present a method that decomposes the problem into (a)~proving that the program
only leaks information it has declassified via \ASSUME annotations already
widely used in deductive program verification;
and (b)~auditing the declassifications against a declarative security
policy.
We show how condition~(a) can be enforced by an extension of the existing program logic \SecCSL,
and how~(b) can be checked by proving a set of simple entailments.
Part of the challenge is to define respective semantic
soundness criteria and
to formally connect these to the logic rules and policy audit.
We support our methodology in an auto-active program verifier, which we apply
to verify the implementations of various case study programs against
a range of declassification policies.
\end{abstract}

\maketitle              


\section{Introduction}
 
Methods for proving that programs are not just functionally correct, but also
maintain confidential information \emph{securely}, have been applied to
realistic software like operating system
kernels~\citep{Murray_MBGBSLGK_13,Costanzo_SG_16}, encompassing features like
concurrency~\citep{Murray_SE_18,Karbyshev+:POST18} and
pointers~\citep{Costanzo_Shao_14,Frumin_KB_21}.  These methods have also been
embodied in auto-active program verification tools like
SecC~\citep{Ernst_Murray_19} and a variant of Viper~\citep{Muller2016,Eilers2018}.
In the \emph{auto-active} verification paradigm, popularized by tools like Dafny~\citep{leino2010usable,leino2010dafny}, VeriFast~\citep{Jacobs11}, and Why3~\citep{filliatre2013why3},
programs are verified semi-automatically via annotations added to their source code,
supporting a high degree practical usability (made possible by advances in automated backend provers like SMT solvers).

The standard criterion for security is
noninterference~\citep{Goguen_Meseguer_82} which guarantees absence of
information leaks.  However, as has been noted
repeatedly~\citep{Sabelfeld_Sands_09,Broberg_Sands_10,vanDelft_HS_15,Broberg_vDS_15,Askarov_Chong_12,Zhang_11,Eggert_vanderMeyden_17,Schoepe_MS_20},
practical programs that handle sensitive information almost always intentionally
reveal some part of that information.  Such an act of \emph{declassification} is
deemed secure if it adheres to a given high-level policy.  As an example
(\cref{sec:overview}), we may release statistics like the average of numbers in
a data set if this set is sufficiently large and/or homogeneous enough such that
the leaked information about the individual data points is acceptably low.  One
approach to reasoning about declassification relies on relational \emph{assume
statements}~\citep{banerjee2008expressive,ChudnovKN14}, and this is used in tools
like SecC~\citep{Ernst_Murray_19} and a variant of Viper~\citep{Muller2016}
---but these tools make no formal connection with high level policy.

A high level policy designates security levels for input and output channels,
with a basic interpretation that observations at a given level should reveal no
information about inputs except those at or below the given level -- and except
for designated intentional releases. Such a declarative policy designates \emph{what}
information may be released to observers at lower level, and \emph{when}, i.e., under what
conditions.  The conditions refer to observable data values and events, like the
size of the data set in our example.  The precise meaning of a declarative
policy can be formalized in terms of observer
\emph{knowledge}~\citep{Askarov_Sabelfeld_07,banerjee2008expressive,BalliuDG11,Askarov_Chong_12,Broberg_vDS_15}.

In a nutshell, this paper contributes an auto-active verification tool that
provably enforces declarative high level policies for concurrent C code, and its
evaluation through challenging case studies.  Security is defined in terms of
knowledge and proved using a relational logic.  Previously, this approach has
only been sketched~\citep[Section~VII, B]{chudnov2018assuming} for a much
simpler program semantics and with no implemented tool.

The \textbf{first contribution} of this paper is to formalize the security
property given by a high level policy, for a programming model with concurrent
threads and dynamically allocated mutable objects, with respect to a \changed{gold}{}
standard threat model.  That is, the property makes the strong guarantee of
\emph{constant-time} security~\citep{almeida2016verifying}, which \changed{not only
forbids secret-dependent branching but also loads and stores to memory addresses
that are secret-dependent.}{precludes secret-dependent branching and loads/stores for memory addresses
that are secret-dependent.}
Thus, even in the presence of variable latency
induced by instruction- and data-caches, a program's running time is independent
of secrets.
\changed{}{This is more restrictive than some security conditions in the literature,
as discussed in \cref{sec:related}, but the threat model is well suited to many context where C programs are used.}

Our formalization disentangles a
\emph{policy-agnostic}~\citep{yang2015preventing} security property from the
\emph{policy-specific} property associated with a high-level declarative policy.
This supports an important methodological point.  In a well designed program,
declassifications occur only at particular places in the code.  To verify a
program, we designate these places by assume statements.  
\changed{The policy-agnostic
property says that no releases occur except those corresponding to the
execution of assume statements.
The policy-specific property says the observer never learns anything except
at such execution steps
where what is released is allowed by the policy,
and only when the associated release condition holds.
}{
The policy-agnostic property says that no releases occur ---i.e., the observer never learns anything--- except at execution steps with assume statements.  
The policy-specific property says that those steps 
only release what is allowed by the policy,
and only when the associated release condition holds.
}

Our \textbf{second contribution} is a deductive proof system which supports
reasoning about assume and assert statements.  A high level policy specifies
conditions under which particular values may be released.  Because the policy is
program-independent, it expresses release conditions as predicates on the
program's I/O history (inspired
by~\citep{banerjee2008expressive,Schoepe_MS_20}).  To verify a program with
respect to a high level policy, each assumption should be justified by an
assertion of a condition in the high level policy that licenses the release
---we call this \emph{policy audit}.  The assertion makes use of a ghost
variable that records the history.  We prove that (a) the proof system is sound
wrt. the policy-agnostic property, and (b) if the program passes the
policy audit then it satisfies the policy-specific security property.
\changed{}{The proofs are mechanized in Isabelle/HOL.}


We choose to base the second contribution on the existing logic \SecCSL~\citep{Ernst_Murray_19},
as it comes not only with a tool implementation but also a mechanized soundness proof that \SecCSL satisfies strong non-interference (absent assume statements).
We extend this framework by uniformly capturing security-relevant semantic
actions of the program (assumptions, actions/outputs, memory access)
as the basis of the policy-agnostic security guarantee. \changed{}{In doing
  so, we inherit from \SecCSL its capabilities for proving security of concurrent
  programs with lock-based synchronization.}

The \textbf{third contribution} of this work is a practical demonstration of the approach.
We have implemented the approach in the auto-active verifier \Tool,
including support for declassification and policy audit.
\Tool treats a subset of C, and is targeted at verifying
concurrent shared-memory programs that use lock-based synchronization.
We leave verification of lock-free programs, exposed to weak-memory
effects~\citep{Yan_Murray_21}, to future work.

We carry out several challenging case studies:
The first, a location service for mobility traces~\citep{10.1007/978-3-319-08506-7_2}
leverages domain knowledge about privacy budgets
wrt. adding \emph{planar laplacian noise}~\citep{Geo_2013}
to achieve differential privacy~\citep{10.1007/11787006_1}
via a suitable policy; verifying that this budget is never exceeded.
The second case study, a sealed-bid auction server ensures
that no client learns anything about the current maximum bid until the auction closes.
Each client's TCP connection is serviced by a separate thread to ensure that no client can block the server's progress and thereby game the auction.
With a variant of this example we demonstrate furthermore
\emph{policy composition}.
The third case study is a verified constant-time implementation of the popular game Wordle,
where rules of the game induce an interesting value-dependent, multi-level
declassification policy
with each move of the player, i.e.,
revealing information about the characters and positions guessed correctly
to that player, but not to other concurrent players nor to the general public.
The final case study considers secure, constant-time, private learning, specifically
differentially private gradient descent to infer a simple linear model,
designed for deployment in a federated learning scenario~\citep{wei2020federated}. 

\Cref{sec:motivation} gives a high-level overview of the ideas.
The threat model is made explicit in \cref{sec:threat-model}.
The semantic and logical foundations are presented in \cref{sec:logic}.
The policy-agnostic and policy-specific guarantees are formalized in \cref{sec:guarantee} and \cref{sec:audit}, respectively.
Case-studies are presented in \cref{sec:case-studies},
and \cref{sec:related} compares to related work and concludes.
\ifExtended
\cref{sec:proofs} sketches \fi proofs for our main results,
which are mechanised in Isabelle/HOL.

\section{Motivation and Overview}%
\label{sec:overview}\label{sec:motivation}

Consider the program in \cref{fig:motivation}, inspired by the \emph{running
  average} example from \citet{Schoepe_MS_20}. It
represents the
state-of-the-art in terms of the size of concurrent programs that
have been verified for secure  declassification prior to \Tool; our
case studies in \cref{sec:examples} are an order of magnitude
larger.

Here, two threads 
cooperate to compute and declassify an aggregate statistic (in this
case a simple average) calculated over purportedly sensitive inputs.
This program's declarative \emph{security policy} allows the
average of the inputs received to be declassified so long as that average
has been calculated over at least 6 inputs.

One thread repeatedly waits for new inputs to arrive
and computes a running sum as well as the count of the inputs received so far;
the second thread reads these values from the shared state,
and prints out the average but only if at least 6 inputs have been counted
in the shared state, to honor the security policy.

\begin{figure}[t]
\begin{lstlisting}
struct avg_state { int count; int sum; };

struct avg_state * avg_lock();
void avg_unlock(struct avg_state *st);

int avg_get_input();
  _(ensures result :: high)
  _(requires %$\History(tr)$%)
  _(ensures  %$\History(\Concat{tr}{\texttt{\bfseries result}})$%)

void print_average(int value);
  _(requires value :: low)

void avg_sum_thread() {
  while(true) {
    struct avg_state * st = avg_lock();
    int i = avg_get_input();
    st->count += 1;
    st->sum += i;
    avg_unlock(st);
  }
}

void avg_declass_thread() {
  struct avg_state * st = avg_lock();
  if (st->count >= 6) {
    int avg = st->sum / st->count;
    _(assume avg :: low) %$\hfill (\dagger)$%
    print_average(avg);
  }
  avg_unlock(st);
}
\end{lstlisting}
\caption{Declassifying the average of at least 6 inputs.\label{fig:motivation}}
\end{figure}

This program makes use of four \emph{external library} functions (those in
\cref{fig:motivation} without an implementation). Such functions are part of
the application's trusted computing base (TCB) and are trusted to be correct and secure.
The external functions \avglock and \avgunlock
implement a mutex that is used to coordinate access to the shared state
(running input count and sum) between the two threads.
Function \avggetinput gets the next input of high sensitivity, i.e., it returns a secret value.
This is specified in terms of its postcondition inside the annotation
\ann{\ensures \result :: \high}, where $e \haslabel \ell$ denotes
that the expression~$e$ holds a value that conforms to classification by security label~$\ell$.
The other two annotations mentioning $\History$ track the input/output history of the program as a trace~$\tr$ of events that are relevant
to the declassification policy, here the trace is extended by the result of the function.
Function \printaverage prints out its argument to a public channel,
its contract similarly mentions a classification, this time,
the precondition inside \ann{\requires value :: \low} specifies
that any argument passed to this function should be of low sensitivity, i.e., public.

What does it mean for this program to be secure?
%
Clearly, the program does not satisfy
noninterference~\citep{Goguen_Meseguer_82}, which requires
that the attacker \emph{never} learns secret values:
since the values stored in \texttt{count} and \texttt{sum} have been computed
from the secret inputs and the attacker will in general learn something
from the call to \printaverage.
Indeed, proving the program secure with vanilla \SecCSL~\citep{Ernst_Murray_19} will not be possible,
therefore we \changed{extend the logic to include}{add} \ASSUME statements in \cref{sec:rules},
which make explicit the decision that the value of \texttt{avg}
can now be considered to be public right before the call to \printaverage.
\changed{}{\cref{fig:motivation} shows use of an assume statement to make
  explicit the declassification ($\dagger$).}

From the point of a verification engineer,
\ASSUME statements express that from this point onwards in the program's execution,
the attacker is assumed to know the declassified value~\citep{chudnov2018assuming},
as with \texttt{avg} in the example above.
And indeed, merely adding assume statements to the logic is easy,
and this is how they have been used in other approaches~\citep{Eilers2018}.
But what justifies such an assumption?
Just as with assuming the absence of hash collisions in cryptographic applications~\citep{DupressoirGJN14},
the act of \emph{declassification should be justified with respect to high level policy.}
The practical challenges in verifying the program
are fundamentally on a different level of abstraction
than the concerns related to formulating and validating declassification policies.
Therefore, we argue, security in the presence of declassification policies
intrinsically suggests separating two concerns:
(a)~the code leaks no information except as made explicit in assumptions,
and (b)~all assumptions are justified by the high level policy.
For both concerns individually we provide a security property
and a sound verification method.


\begin{quote}
\textbf{Goal: policy-agnostic security guarantee.} 
Information leaks in a \emph{verified} program
can always be traced back to a prior failed assumption.
\end{quote}
To make this intuition precise we consider the attacker's knowledge at a given
point in an execution (called the ``major run''~\citep{Beringer12}), in terms of
their \emph{uncertainty}~\citep{Askarov_Sabelfeld_07,chudnov2018assuming,Schoepe_MS_20} about the
initial secrets.  The uncertainty is the set of runs that are consistent with
what the attacker is able to observe about the major run.  A given ``minor run''
gets removed from the uncertainty at any step of the major run where the
attacker can observe something inconsistent with that minor run.  Our
formalization is based on the notion of a \emph{schedule}, a generic semantic
model that relates such observations to \changed{the execution of those program locations that make}{points in control flow annotated with explicit} assumptions about attacker knowledge.
In \cref{fig:traces} the
attacker is uncertain about the possible initial state~$s'_1$ because there
exists a minor run up to the state~$s'_j$ whose schedule (primed~$\sigma'$) is
observationally equivalent, written $\approx_\ell$, to the schedule of the major
run (unprimed~$\sigma$).  In the $j$-th step, an information leak occurs when
the outputs mismatch $v \neq v'$.  The policy-agnostic security guarantee ensures
that this mismatch can always be explained by an assumption~$\rho$, occurring at an
earlier $i$-th step, that is not satisfied at that point (cf. red backwards
arrow).  Dually, the attacker is not allowed to exclude runs from their
uncertainty that have no such assumption violation.  We formally prove this as
\cref{thm:guarantee} in \cref{sec:guarantee}.

Note that formulas $\rho$ like $e \haslabel \ell$ are interpreted \underline{r}elationally over pairs of states, specifically $e \haslabel \low$ means that~$e$ evaluates to the same value in both.
In the example from \cref{fig:motivation}, therefore,
the assumption $\rho \mathrel{\defeq} \texttt{avg} \haslabel \low$ \changed{}{(marked by ($\dagger$))} must have failed between some prior $s_i$ and $s'_i$,
such that the difference between $v$ and $v'$ is \emph{caused} by the different values of \texttt{avg}.

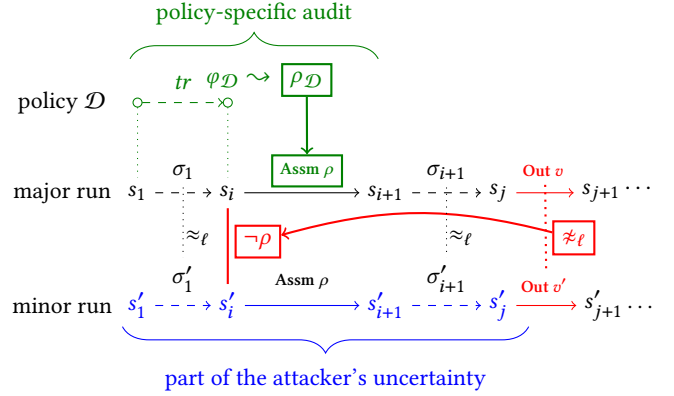
\begin{figure}[t]
    \centering
\begin{tikzpicture}
    \node (s0) at  (0.0,1.5) {major run};
    \node (s1) at  (1.0,1.5) {$s_1$};
    \node (s2) at  (2.2,1.5) {$s_i$};
    \node (s3) at  (4.3,1.5) {$s_{i+1}$};
    \node (s4) at  (5.8,1.5) {$s_j$};
    \node (s5) at  (7.4,1.5) {$s_{j+1} \cdots$};
    \path[draw,dashed,->] (s1) -- node[above=2pt] (sigma1) {$\sigma_1$}
                          (s2);
    \path[draw,->] (s2) -- node[above=2pt,black!50!green,xshift=2pt,draw,thick] (a1) {\scriptsize $\Assumption{\rho}$} (s3);
    \path[draw,dashed,->] (s3) -- node[above=2pt] (sigma2) {$\sigma_{i+1}$}
                          (s4);
    \path[draw,red,->] (s4) -- node[above=3pt,red] (so) {\scriptsize $\Out{v}$} (s5);

    \node (t0) at  (0.0,0) {minor run};

    \node (t5) at  (7.4,0) {$s'_{j+1} \ldots$};

    {\color{blue}
    \node (t1) at  (1.0,0) {$s'_1$};
    \node (t2) at  (2.2,0) {$s'_i$};
    \node (t3) at  (4.3,0) {$s'_{i+1}$};
    \node (t4) at  (5.8,0) {$s'_j$};
    \path[draw,dashed,->] (t1) -- node[above=2pt] (tigma1) {\normalcolor $\sigma'_1$}
                          (t2);
    \path[draw,->] (t2) -- node[above=3pt] {\scriptsize $\normalcolor \Assumption{\rho}$} (t3);
    \path[draw,dashed,->] (t3) -- node[above=2pt] (tigma2) {\normalcolor $\sigma'_{i+1}$}
                          (t4);
    }

    \path[draw,red,->] (t4) -- node[above=2pt,red] (to) {\scriptsize $\Out{v'}$} (t5);

    {\color{red}
    \node[draw,thick,red,rectangle] (rho) at (2.6,0.85) {$\lnot\rho$};
    \node[draw,thick] (neq) at (6.8,0.9) {$\not\approx_\ell$};
    \path[draw,thick] (s2) -- (t2);
    \path[draw,thick,dotted] (so) -- (to);
    \path[draw,thick,->] (neq) to[bend right=15] (rho);
    }

    \node at (1.8,0.9) {$\approx_\ell$};
    \node at (5.3,0.9) {$\approx_\ell$};
    \path[draw,thin,dotted] (sigma1) -- (tigma1);
    \path[draw,thin,dotted] (sigma2) -- (tigma2);

    \draw[decorate,blue,decoration={brace,amplitude=10pt,mirror},yshift=-8pt]
          (0.8,0) -- (6.2,0) node[black,midway,yshift=-0.7cm]
          { \color{blue} part of the attacker's uncertainty};

    \node (h0) at (0.0,2.7) {policy $\D$};

    {\color{black!50!green}
    \node[draw,circle,scale=0.4] (h1) at (1.0,2.7) {};
    \node[draw,circle,scale=0.4] (h2) at (2.2,2.7) {};
        \node at (2.35,3.0) {$\Dpre \leadsto$};
        \node[draw,thick] (Dpost) at (3.25,3.0) {$\Dpost$};
    \path[draw,dotted] (h1) -- (s1);
    \path[draw,dotted] (h2) -- (s2);
    \path[draw,->,dashed] (h1) -- node[above=2pt] {$\tr$}  (h2);
    \path[draw,->,thick] (Dpost) -- (a1);

    \draw[decorate,decoration={brace,amplitude=10pt},yshift=8pt]
          (0.9,3) -- (4.2,3) node[black,midway,yshift=0.6cm]
          { \color{black!50!green} policy-specific audit};
    }
\end{tikzpicture}
\caption{Visual representation that relates a major run to a hypothetical
         minor run, where the mismatch in outputs can be traced to an earlier failed assumption.}
\label{fig:traces}
\end{figure}

A declassification policy~$\D(\tr)$ is formulated
over a representation of the \emph{execution trace}~$\tr$,
that is collected as auxiliary ``ghost'' state in the verification
as a syntactic expression of a suitable sequence data type.
The trace is tracked with an abstract I/O predicate~$\History(\tr)$
that specifies that the trace denoted by expression $\tr$
is the current history of the program at that point
and that $\tr$ is well-formed wrt. the application (examples in \cref{sec:case-studies}).



A policy $\D(\tr)$ over trace~$\tr$ takes the form
$\policy{\Dpre(\tr)}{\Dpost(\tr)}$.
It  consists of a \emph{condition}~$\Dpre$ on traces 
which states \emph{when} a declassification is permitted,
and a \changed{}{relational} \emph{release formula}~$\Dpost$,
which encodes \emph{what} information is allowed to be released.
For the example, the policy requires that the trace~$\tr$ has a length of at least 6;
it then justifies to classify the average over the numbers stored in~$\tr$ as~$\low$:
\begin{align*}
\D(\tr) \quad \defeq \quad \policy{\length(\tr) \ge 6~}{~\sumfunc(\tr)/\length(\tr) \haslabel \low}
\end{align*}
Here, the critical issue is to enforce whether all assumptions
made in the program are covered by the policy,
with respect to the symbolic path constraints at that point, 
which we call an \emph{audit} (exemplified below).
\begin{quote}
\textbf{Goal: certified adherence to policy.}
In a program that has been \changed{}{correctly} \emph{audited} 
with respect to a declassification policy,
the information leak associated with each assumption
is bounded by the policy.
\end{quote}
In \cref{fig:traces} this is depicted at the top (green):
each possibly failing assumption has to be justified from the policy~$\D$
with respect to the trace prefix~$\tr$ at that point.
Together with the policy-agnostic guarantee,
this implies that from the execution of a verified and audited program an attacker can only gain knowledge that is allowed by the policy.
We formalize this as \cref{thm:audit} in \cref{sec:audit}
by integrating a declarative semantics of policies over traces
with the program's schedules, i.e., with the ground-truth about the execution.


To prove the audit (i.e.\ that each assumption that introduces a leak is justified by the declassification policy~$\D$) we make use of
\emph{invariants} established by the verification
that connect the program's state to the abstract trace.
For the example, this means that the verification attaches the
following resource invariant to the shared state (acquired when calling
\texttt{avg\_lock()} and required to be true when calling \texttt{avg\_unlock()}):
\begin{align}
\inv{\tr,\texttt{st}}
    \quad \defeq \quad
        \History(\tr)    & \land \texttt{st->count} = \length(\tr)
            \label{eq:inv} \\
                {} & \land \texttt{st->sum} = \sumfunc(\tr)
                \nonumber
\end{align}

Together with path condition \texttt{st->count >= 6} from the \texttt{if} test,
the audit formally certifies that right before the assume, the policy condition $\Dpre(\tr)$ is entailed. Auditing then requires us to prove that the
policy release formula~$\Dpost(\tr)$ (what the policy allows to be declassified)
implies the assumption \texttt{avg} \haslabel \low (what the program
has actually declassified), which holds under the resource invariant and
thus this example satisfies secure declassification.
\ifExtended
\changed{}{The full proof for this example appears in \cref{app:motivation-proof}.}
\fi

In summary, successful verification of an annotated program ensures that every increase in attacker knowledge occurs following an assumption that allows the new knowledge.  
Successful policy audit ensures that for any assumption there is a 
policy clause $\policy{\Dpre(\tr)}{\Dpost(\tr)}$
such that the path condition $P$ at the assumption implies the release condition $\Dpre(\tr)$
which in turn validates $\Dpost(\tr)$, i.e., the released information is allowed by the policy.

%

\section{Threat Model}%
\label{sec:threat-model}

We assume a program being verified that contains a number of top-level
functions like \avgdeclassthread and \avgsumthread in \cref{fig:motivation}.
These top-level functions are necessarily invoked from unverified, trusted
\emph{wrapper} code like \texttt{main()} that may invoke multiple
instances of each function and in parallel.
Verified application functions also make use of external library functions
like \printaverage that are not to be verified and must be trusted.

Our approach rests on a number of assumptions, including~\citep{Murray_vanOorschot_18} \emph{adversary expectations}---those assumptions that apply to the attacker---and \emph{domain hypotheses}---those that apply to the program's environment. We also make various meta-assumptions
that apply to the verification approach itself.

\myparagraph{Adversary Expectations.} We assume
a passive attacker with arbitrary security level~$\ell$
who knows the verified program's source code and its proof (as expressed in the annotations in the source). We assume that the attacker can
observe all data passed to any external library function whose precondition requires
that data to be $\ell$ or below (including \low).
Similarly, we assume that the attacker initially
knows only those program data that have been either marked at $\ell$ or below in the precondition of a top-level application function, or in the postcondition of an external library function.
Without loss of generality, we assume that all secrets are contained somewhere in the initial program state (e.g.~via a standard, deterministic oracle semantics~\citep{ClarkH08,Murray_MBGBSLGK_13} for functions like \avggetinput, not elaborated here).
Following the \emph{constant-time} security threat model~\citep{almeida2016verifying}, 
we assume
the attacker can observe not only the
program's execution time but also its memory access pattern and which
conditional branches it takes. These latter are made observable via timing effects induced by microarchitectural elements like caches. The attacker can also
observe the
program's concurrent schedule~\citep{Ernst_Murray_19,Frumin_KB_21}.

\myparagraph{Domain Hypotheses.}
We assume that the unverified code is correct and secure:  any preconditions of top-level verified application functions are
always satisfied whenever they are invoked by non-verified code, and
all external library functions
satisfy their contracts. We assume
that the verified program executes faithfully atop an
operating system whose scheduler does not leak sensitive information~\citep{Murray_MBGBSLGK_13} and is insulated from transient
execution effects~\citep{canella2019systematic}.

\myparagraph{Meta Assumptions.}
\changed{We assume that \Tool is correctly implemented.}
{\Tool implements the logic presented in this paper (\cref{sec:logic} with the extensions from \cref{sec:audit}), not for the simple command language defined in \cref{sec:logic} over which our soundness theorems are proved,
but for a substantial fragment of the C language.
We assume that this implementation is faithful to the theoretical ideas of this paper
as well as to the semantics of the supported part of  C (cf. \cref{sec:case-studies}).}
A small caveat is that \Tool currently models signed integers as
mathematical integers, so we assume that traditional verification methods
have also been applied to ensure absence of overflow, which is an orthogonal
problem to those considered by this work.
\changed{}{Since \Tool relies on Z3 to discharge verification conditions,
we assume that Z3 has no soundness bugs that affect our proofs.}
These assumptions are common to auto-active verifiers.
\changed{}{Similarly, for the mechanized proofs,
we rely on the soundness of Isabelle/HOL,
which thanks to its small-kernel architecture is highly trustworthy.}
%

\myparagraph{Claim.}
Then we claim that if verified with
\Tool, the program in question will adhere to its security policy as
specified, against the aforementioned attacker. Specifically this attacker
can learn no information other than what the program declassifies, and
all declassifications are in accordance with the security policy.

\section{Background: \SecCSL}\label{sec:logic}

\changed{}{Security Concurrent Separation Logic} (\SecCSL)
is a program logic proposed in \citep{Ernst_Murray_19}
for proving timing-sensitive noninterference of concurrent programs.
\changed{}{It extends Concurrent Separation Logic~\citep{OHearn_04}
by adding new assertions to capture security-related properties
and by adapting the existing proof rules to ensure that these are maintained soundly.}
We adapt and extend \SecCSL as the foundation of \Tool, for its native
support for compositional, modular, implementation-level reasoning.

\changed{[Moved/adapted from below:] }{
Judgements in the logic have the form
\[ \prove{\ell}{P}{c}{Q} \]
for a command~$c$ and pre-/postcondition $P$~and~$Q$,
where $\ell$ denotes the level of the attacker, typically $\low$.
It implies semantically that if the program~$c$ is executed when precondition~$P$ holds then, $c$'s
execution will be memory-safe
and when $c$ terminates the postcondition $Q$ will hold
(partial correctness).
Moreover, with respect to the adversary assumptions stated in \cref{sec:threat-model},
the execution will not leak information
to the~$\ell$-level attacker.
This includes via \emph{timing channels} as
the proof rules enforce that programs never branch on secrets nor perform secret-dependent memory accesses.
}

In this section we \changed{first}{\!\!} present the necessary background
on the \changed{syntax, }{assertion language,} the proof rules,
\changed{}{and the semantic foundations.}
\changed{and then 
explain how the contributions of this paper are integrated.}{
We rely on these preliminaries in \cref{sec:guarantee}
and \cref{sec:audit}, where we will discuss how the guarantees
entailed by $\prove{\ell}{P}{c}{Q}$ are formalized---these guarantees
are the key difference between our work and \SecCSL.}

The logic is defined for a core programming language
with commands~$c$:
\begin{align}
 c \; ::= \hspace*{1em} & 
         \ASSIGN{x}{e}
    \mid \LOCK{l}
    \mid \UNLOCK{l}
    \mid  c_1 ; c_2
    \mid  c_1 \parallel c_2 \mid \notag \\
    & 
         \ITE{\phi}{c_1}{c_2}
    \mid \WHILE{\phi}{c} \mid \notag \\
    & 
         \ASSUME{\rho}
    \mid \OUTPUT{\el}{\ev}
    \mid \LOAD{x}{\ep}
    \mid  \STORE{\ep}{\ev} \mid
        \tag{$\dagger$}
        \\
    & 
    \TRACE{e}
        \tag{$\ddagger$}
\end{align}
The details of most commands present in \SecCSL already remain unchanged:
this includes assignments, locking, sequential/parallel composition,
conditionals, and while loops.
\changed{We explain their semantics and proof rules in \cref{sec:assume}.}{}

\changed{}{\textbf{Extensions to \SecCSL: }}
Outputs, assumptions, and $\TRACE{e}$ are new wrt. \SecCSL.  The two commands
for memory load and store are modeled differently here to capture constant-time security.
The commands on the third line~($\dagger$) are relevant to the program's
policy-agnostic security guarantee
\changed{with respect to the adversary assumptions stated in \cref{sec:threat-model}}{as discussed in \cref{sec:guarantee}}.
The command~$\TRACE{e}$~($\ddagger$) represents the occurrence of application-specific events with data~$e$,
\changed{such as the return value of \avggetinput,}
{such as calls to \avggetinput where $e$ captures its return value.}
\changed{that}{These events} are not visible to the attacker but instead give a declarative
account of the guarantees entailed by policy enforcement,
\changed{}{as explained in \cref{sec:audit}}.
\changed{In preparation for our two main theorems in \cref{sec:guarantee}
and \cref{sec:audit},
in \cref{sec:property}
we present a strong auxiliary property over pairs of runs
that is induced by the proof rules.}{}

\smallskip

\noindent \textbf{Notation:} For a schedule~$\sigma$,
formalized as a mathematical sequence of actions (defined later),
we write $|\sigma|$ to denote its the length, $\Concat{\sigma_1}{\sigma_2}$ denotes concatenation, $\List{a}$ is the singleton sequence with element~$a$,
$\sigma_i$ is the $i$--th entry if $i < |\sigma|$ and $\sigma\cut{i}$ is the prefix of length~$i$.
Program states make use of stores~$s$,
modeled as mathematical maps; we write~$s(x)$ for lookup
and $s(x:=v)$ for map override.
Heaps~$h$ in addition have a domain~$\dom{h}$,
and we write \changed{}{$h = [a \mapsto v]$ for a singleton heap that maps address~$a$ to value~$v$ and} $h_1 \uplus h_2$ for the union of two heaps,
implying that they need to have disjoint domains.

\changed{
\subsection*{Background and Preliminaries}%
\label{sec:seccsl}}{}

\newcommand{\pointsto}[1]{\stackrel{#1}{\mapsto}}

\changed{\SecCSL is a program logic proposed by \citet{Ernst_Murray_19}
for proving timing-sensitive noninterference of concurrent programs.
It extends Concurrent Separation Logic~\citep{OHearn_04}
by adding new assertions to capture security-related properties
and by adapting the existing proof rules to ensure that these are maintained soundly.}{}


\subsection{Expressions and Assertions}

The typed language of expressions~$e$ 
includes boolean formulas~$\phi \colon \mathit{Bool}$
and a designated sort $\mathit{Label}$ for security labels,
including at least the constants $\low$ and $\high$
and a binary relation $\sqsubseteq$ that satisfies the lattice axioms.
The assertion language to formulate pre-/postconditions
includes the standard connectives, quantifiers, and separation logic primitives
points-to and separating conjunction:
\begin{align*}
& \text{assertion } P \ ::= \\
&        \quad  \phi
    \mid e \haslabel \el
    \mid  \Emp
    \mid \pto{\ep}{\ev}
    \mid  P_1 \star P_2
    \mid  P_1 \implies P_2
    \mid  \Exists{x}{P} 
    \mid \cdots 
\end{align*}
In addition, we have \emph{value sensitivity} or classification $e \haslabel \el$,
which expresses that the value of~$e$ is safe to be known for an $\el$ attacker,
where $e \haslabel \low$ is the strongest such assertion
and $e \haslabel \high$ is just true.
Expressive power comes from the reflection of the security lattice
into the assertion language, such that labels~$\ell$ are symbolic expressions, too.
For example, $e \haslabel (d \mathrel{?} \high : \low)$ denotes a classification of~$e$
conditional on the current value of~$d$, where $(\_ \mathrel{?} \_ : \_)$ is an if-then-else expression.
\changed{}{Spatial assertions include the empty heap $\Emp$,
the points-to predicate $\pto{\ep}{\ev}$,
i.e., $\ep$ points to valid memory containing value $\ev$,
and the separating conjunction $P_1 \star P_2$,
i.e., assertions $P_1$ and $P_2$ hold on disjoint parts of the heap, respectively,
as usual~\citep{reynolds2002separation,OHearn_04}.}

An assertion is called \emph{pure} if it does not make reference to the heap,
and it is called \emph{relational} if it includes a classification.
We denote pure relational formulas by letter $\rho$ in the following.

Semantically, expressions~$e$ are evaluated over a store~$s$,
written~$\sem{e}{s}$ \changed{(the details do not matter here)}{in the standard way}.
Assertions, in contrast, are evaluated over a \emph{pair of states},
each consisting of a store~$s$ and a heap~$h$,
with respect to the attacker level~$\ell$.
We write~$s,s' \models_\ell \rho$ when pure assertion~$\rho$ holds
and~$(s,h),(s',h') \models_\ell P$ that spatial assertion~$P$ holds,
where intuitively~$s$ and~$h$ are taken from the actual ``major'' run of the program that is compared to~$s'$ and~$h'$ from
some hypothetical ``minor'' run in Beringer's terminology~\citep{Beringer12} (cf. \cref{fig:motivation}).

\begin{figure}
\begin{align*}
(s,h),(s',h') \models_\ell \phi
    &\quad\defeq\quad s,s' \models_\ell \phi \land h=h'=\emptyheap
            \quad\text{for } \phi \text{ pure } \\
\text{where} \quad s,s' \models_\ell \phi
    &\quad\defeq\quad \sem{\phi}{s} \land \sem{\phi}{s'}
    \\[4pt]
s,s' \models_\ell e \haslabel \el
    &\quad\defeq\quad \sem{\el}{s} \sqsubseteq \ell \land \sem{\el}{s'} \sqsubseteq \ell
        \implies \sem{e}{s} = \sem{e}{s'}
\end{align*}
\begin{added}
\begin{align*}
(s,h),(s',h') \models_\ell \pto{\ep}{\ev}
    &\quad\defeq\quad h = [\sem{\ep}{s} \mapsto \sem{\ev}{s}] \\
    &\rlap{\hspace{4pt} \text{and}\ } \phantom{\quad\defeq\quad}
        h' = [\sem{\ep}{s'} \mapsto \sem{\ev}{s'}] \\[4pt]
(s,h),(s',h') \models_\ell P_1 \star P_2
    &\quad\defeq\quad h = h_1 \uplus h_2 \text{ and } h' = h'_1 \uplus h'_2 \\
    &\rlap{\hspace{4pt} \text{and}\ } \phantom{\quad\defeq\quad}
        (s,h_i),(s',h_i') \models_\ell P_i \text{ for } i=1,2 \\[4pt]
(s,h),(s',h') \models_\ell P_1 \implies P_2
    &\quad\defeq\quad
(s,h),(s',h') \models_\ell P_1 \\
    &\hspace{1cm}\text{implies } (s,h),(s',h') \models_\ell P_2
\end{align*}
\begin{align*}
(s,h),(s',h') \models_\ell \exists x.\ P
    &\quad\defeq\quad
        \text{there are } v,v' \text{ with } \\
    &\phantom{\quad\defeq\quad}
        (s(x:=v),h),(s'(x:=v'),h') \models_\ell P
\end{align*}
\end{added}
    \caption{\changed{}{Relational semantics of assertions.}}
    \label{fig:sem-assert}
\end{figure}

\changed{}{The semantics of assertions is shown in \cref{fig:sem-assert}.}
Non-relational formulas~$\phi$ must hold in both states individually.
Value sensitivity $e \haslabel \el$ enforces
agreement of the values $\sem{e}{s}$ and $\sem{e}{s'}$ in both states
for an attacker at level~$\ell$ that is at least as high as $\el$.
Note that this semantics improves on original \SecCSL~\citep{Ernst_Murray_19} by
not requiring $\el$ to agree, which is crucial for specifying the Wordle
security policy (\cref{sec:wordle}). 
As in \SecCSL~\citep{Ernst_Murray_19}, pure assertions impose an empty heap.
\changed{whereas spatial assertions
extend standard Separation Logic semantics to pairs of states point-wise. }%
{Spatial assertions extend standard Separation Logic (SL) semantics to pairs of states point-wise.}

In contrast to the grammar shown, the original \SecCSL supports in addition
\emph{memory location sensitivity},
written $e_p \pointsto{\el} e_v$,
which expresses that both memory \emph{access} via address~$e_p$ is observable
to an $\el$-attacker as well as the actual value stored therein.
However, this feature comes with some trade-offs,
e.g. assertions are restricted to the positive fragment of SL
and logical entailment becomes somewhat complex.
In \cref{sec:rules} below we offer a different approach based on the standard points-to
assertion that avoids these limitations.


\subsection{Proof Rules}\label{sec:rules}

The proof rules of \SecCSL to derive judgements $\prove{\ell}{P}{c}{Q}$
are all ``matched'' (aka ``synchronous'') rules, where the control flow of program~$c$
is always the same between the major and minor run.
The approach is adequate for \emph{timing-sensitive} security
as discussed in \cref{sec:threat-model}
and integrates nicely into existing logics
\changed{at the expense of not being able to branch on secrets.
}{(though it is not adequate for weaker security properties that allow some branching on secrets).}
\ifExtended
\changed{}{\cref{app:motivation-proof} shows the proof sketch built from these rules
  for the example from \cref{sec:motivation}.}
\fi

\begin{figure*}[t]
  \begin{mathpar}
    \changed{}{\infer{\ }{\prove{\ell}{P(e)}{\ASSIGN{x}{e}}{P(x)}}} 

    \changed{}{\infer{\ }{\prove{\ell}{P}{\LOCK l}{P \star \inv{l}}}}

    \changed{}{\infer{\ }{\prove{\ell}{P \star \inv{l}}{\UNLOCK l}{P}}}

    \changed{}{
        \infer{\prove{\ell}{P_1}{c_1}{P_2} \and
               \prove{\ell}{P_1}{c_2}{P_2}}
              {\prove{\ell}{P_1}{c_1;c_2}{P_2}}}

    \changed{}{
        \infer{\prove{\ell}{P_1}{c_1}{Q_1} \and
               \prove{\ell}{P_1}{c_2}{Q_2} \and
               \fv{c_i} \cap \mathrm{mod}(c_j) = \varnothing}
              {\prove{\ell}{P_1 \star P_2}{c_1 \parallel c_2}{Q_1 \star Q_2}}}

    \infer{\seccheck{P \implies \phi\haslabel\ell} \and \prove{\ell}{P \star \phi}{c_1}{Q} \and \prove{\ell}{P \star \lnot \phi}{c_2}{Q}}%
          {\prove{\ell}{P}{\ITE{\phi}{c_1}{c_2}}{Q}} 

    \infer{\seccheck{P \implies \phi\haslabel\ell} \and \prove{\ell}{P \star \phi}{c}{P}}%
          {\prove{\ell}{P}{\WHILE{\phi}{c}}{P \star \lnot \phi}} 

    \changed{}{
    \infer{\seccheck{P \implies \phi\haslabel\ell} \and \prove{\ell}{P\ \star \phi}{c}{Q} \and \prove{\ell}{P\ \star \lnot \phi}{c}{Q}}%
          {\prove{\ell}{P}{c}{Q}} \textsc{Split}}
  \end{mathpar}
  \caption{Some rules of \SecCSL, where $\fv{e}$ denotes free variables of~$e$.
  \label{fig:rules}
  \seccheck{\textbf{Highlighted}} side-conditions represent security checks.
    }
\end{figure*}

The proof rules of \SecCSL work exactly like those of traditional Concurrent SL~\citep{reynolds2002separation,OHearn_04}
except that the rules for $\textbf{if}$ and $\textbf{while}$ in \cref{fig:rules}
\changed{as well as logical case splits}{}
enforce that the respective branch condition~$\phi$ is not secret from the view of the $\ell$-attacker.
\changed{}{A similar requirement applies to \emph{logical} case splits (rule \textsc{Split}): semantically, there are four combinations how~$\phi$
could evaluate, of which our assertion language can represent two (cf. second line in \cref{fig:sem-assert}).}
These side conditions related to security are highlighted in \seccheck{\mathrm{light\ gray}}.
 \changed{}{
Rules for sequential and parallel composition, acquiring and releasing locks, as well as the frame and consequence rules
are entirely standard, see e.g. \citep{reynolds2002separation} (sequential fragment) and~\citep[Fig.~2]{gotsman2011precision} (concurrency and locks).
The locking rules transfer ownership of a \emph{lock invariant}
denoted $\inv{l}$ for lock~$l$ such as the one in \cref{eq:inv}.}

\begin{figure}[t]
    \begin{mathpar}
    \infer{\ }{\prove{\ell}{\Emp}{\ASSUME{\rho}}{\rho}} \textsc{Assume}

    \infer{\ }{\prove{\ell}{\seccheck{\ev \haslabel \el \star \el \haslabel \ell}}{\OUTPUT{\el}{\ev}}{\Emp}} \textsc{Output}


    \infer{x \notin \fv{\{e,\ep\}}}{\prove{\ell}{\pto{\ep}{e} \ \seccheck{\star\  \ep :: \ell}}{\LOAD{x}{\ep}}{x = e \star \pto{\ep}{e}}} \textsc{Load}

    \infer{\ }{\prove{\ell}{\pto{\ep}{e} \ \seccheck{\star\  \ep :: \ell}}{\STORE{\ep}{\ev}}{\pto{\ep}{\ev}}} \textsc{Store}
    \end{mathpar}
    \caption{Proof rules for commands that produce
             security-relevant actions; security checks are \seccheck{\textbf{highlighted}}.
             Discussion of command $\TRACE{e}$ is deferred to \cref{sec:audit}.}
    \label{fig:ioa}
\end{figure}

The proof rules for \changed{the~}{atomic} commands except $\TRACE{e}$ are shown in \cref{fig:ioa}. 
Rule \textsc{Assume} just manifests the assumed formula~$\rho$ to the postcondition\changed{}{~\citep[Sec. VII B]{chudnov2018assuming}}.
In comparison to the other three rules, there is no justification yet
why this assumption can be made---as described in \cref{sec:motivation},
this justification comes from the global declassification policy instead
as formalized in \cref{sec:audit}, \changed{}{where we will also show the rule for command $\TRACE{e}$}.

The rule for commands $\OUTPUT{\el}{\ev}$, which \changed{recall}{} outputs the value~$\ev$
to the attacker who is at security level~$\el$, requires that the
$\el$-level attacker knows the value~$\ev$ being output and, hence, does not
learn any new information. It also requires that the expression
$\el$ denoting the level at which~$\ev$ is being output is known to
the $\ell$-level attacker since, otherwise, the choice of the level on which
the output is occurring could leak information.
In the example from \cref{sec:motivation}, an output command is represented as library function \printaverage, which specifies that its argument must be $\low$.

In accordance with the threat model from \cref{sec:threat-model},
the rules for loading and storing via pointer~$\ep$ are similarly guarded
to not leak information via the memory access pattern,
by enforcing that the pointer is not sensitive.

\subsection{Program Semantics}\label{sec:sem-cmd}
\changed{We briefly sketch how semantics of program execution is formulated
as we rely on this and extend it later.}
{We briefly sketch how program execution is
modeled by a typical small-step operational semantics,
similarly to~Vafeiadis' formulation~\citep{vafeiadismfps11},
as we rely on this and extend it later.}
A \emph{configuration} captures the runtime state of a program:
\changed{}{
\begin{align*}
\text{configuration }
    k,k' ::= \Run{L}{c}{s}{h} \mid \Stop{L}{s}{h} \mid \Abort
\end{align*}
}
\changed{There are three kinds of semantic configurations~$k$:}{}
The configuration
$\Run{L}{c}{s}{h}$ represents a running program whose current state is given
by the store~$s$ and heap~$h$ and whose remaining program to execute is the
command~$c$; $L$ is the set of locks not currently acquired. The
configurations $\Stop{L}{s}{h}$ and $\Abort$ represent respectively the (successfully)  terminated
program whose final unacquired locks are~$L$ and whose final state is the
store~$s$ and heap~$h$, and aborted programs (e.g. due to a memory violation).
\SecCSL associates to each lock a resource invariant,
\changed{}{like the one shown in \cref{eq:inv}} in \cref{sec:motivation}.
We denote by $\invs{L}$ the conjunction of the invariants
of locks in the set~$L$, i.e., shared state that is
not currently accessed in a critical section.

The transition relation $\Step{k}{\sigma}{k'}$
represents one execution step from configuration~$k$ to configuration~$k'$ producing the \emph{schedule}~$\sigma$.
Schedules~$\sigma$ record the sequence of
\emph{actions} of the program that are relevant to
capture the observational powers of the attacker under the threat model
of \cref{sec:threat-model}.
We denote by $\Steps{k}{\sigma}{k'}$ the reflexive transitive closure of the transition relation.
The action~$\Fst$ (respectively~$\Snd$) represents
the decision to schedule the left (respectively right) command in a parallel
composition $\parallel$.
The action $\Tick$ represents the execution of atomic command like assignments.
In the next sections we will extend this semantic model of actions
to capture the program's input-/output behavior,
those steps that correspond to assumptions made in the verification,
and we will also make memory access explicit in the schedule.

\changed{}{Our security guarantees will be expressed relative to what an attacker can observe
from the schedules of runs and whether information leaks are covered by policies.
In \cref{sec:guarantee} and \cref{sec:audit} we will capture these notions formally.}
\changed{}{
  \ifExtended
  The complete set of semantic rules is in \cref{app:sem}, \cref{fig:sem-cmd}.
  \fi
    In the next section we present those that are relevant to the our extensions of \SecCSL.}

While the noninterference guarantee of \SecCSL~\citep[Theorem~2]{Ernst_Murray_19}
focuses on comparing heap locations, we point out that
$\prove{\ell}{P}{c}{Q}$ implies that for a given major execution
$\Steps{\Run{L}{c}{s}{h}}{\sigma}{k}$
of program~$c$ ending in some final/intermediate configuration~$k$,
any minor run
$\Steps{\Run{L}{c}{s'}{h}}{\sigma'}{k'}$
of the same length, as encoded by $|\sigma| = |\sigma'|$,
will satisfy that the attacker learns nothing from the schedule,
which we relax to equivalence of observations below, written $\sigma \approx_\ell \sigma'$, after introducing additional types of actions into the schedule.
Moreover, the matched rules enforce semantically that~$k$ and~$k'$ are either both final
or both running configurations with the same program.
\changed{We state a corresponding result in \cref{sec:property}
as a step towards our main results in \cref{sec:guarantee}
and \cref{sec:audit}.}{}



\section{Policy-Agnostic Guarantee}%
\label{sec:guarantee}

\changed{}{In this section, we discuss the \emph{policy-agnostic} part of our contribution.
It is based on the notion of semantic \emph{actions},
which give rise to \emph{schedules} that are much more informative than those of vanilla \SecCSL,
which in turn allows us to reason about attacker knowledge gained from observing
an execution of a program that exhibits such schedules (\cref{defn:visible,defn:uncertainty}).}

\begin{deleted}
\subsection*{A Semantic Model of Actions and Schedules}%
\label{sec:assume}

In this section we discuss the fragment of the language~($\dagger$)
and~($\ddagger$), \changed{}{i.e., our additions to \SecCSL} that deal with outputs, assumptions, (constant-time) memory access, and trace events:
\begin{align*}
& \text{command } c ::= \\
& \quad \cdots \mid \ASSUME \rho \mid \OUTPUT{\el}{\ev}
    \mid \LOAD{x}{\ep} 
    \mid \STORE{\ep}{\ev} 
    \mid \TRACE{e}
\end{align*}
Multiple output channels are modeled in this formalization
by attaching a security label $\el$ to the output command,
so that output~$\ev$ is visible only to an attacker with that observational capability.

\myparagraph{Program Semantics (\cref{fig:sem}).}

The commands in this section have in common
that they expose an \emph{action} as part of the schedule over which
we can express and relate the different aspects of security.
\end{deleted}

The grammar of actions is as follows
\[ \text{action } a ::= \tau \!\mid\! \Fst \!\mid\! \Snd 
            \!\mid\! \Out{\ell}{v} 
            \!\mid\! \Assumption{s~\rho}
            \!\mid\! \Load{p}
            \!\mid\! \Store{p}
            \!\mid\! \Trace{e} \]
where internal action~$\tau$ and concurrent scheduling decision $\Fst$, $\Snd$
are inherited from \SecCSL (cf. \cref{sec:sem-cmd}), and the new actions
arise from the execution steps of the commands discussed in this section.
The key transitions are shown in \cref{fig:sem}.
Heap access through an invalid pointer instead produces an $\Abort$ successor configuration (but still exposes the same action in the schedule).
The rule for memory writes (stores) is analogous to that of loads.
Formulated this way, the extra provisions or security in terms of the action labels
do not in any way constrain the program execution,
it just exposes the necessary information for a later analysis.

\begin{figure}
\begin{align*}
\Run{L}{\ASSUME{\rho}}{s}{h}
& \xrightarrow{~\List{\Assumption{s~\rho}}~}
    \Stop{L}{s}{h} \\[8pt]
\Run{L}{\OUTPUT{\el}{\ev}}{s}{h}
& \xrightarrow{~\List{\Out{\sem{\el}{s}}{\sem{\ev}{s}}}~}
    \Stop{L}{s}{h}
    \\[8pt]
\Run{\LOAD{x}{\ep}}{L}{s}{h}
& \xrightarrow{~\List{\Load{\sem{\ep}{s}}}~}
    \Stop{L}{s'}{h} \\
    \text{for } & s' = s\big(x:=h(\sem{\ep}{s})\big) \text{ if } \sem{\ep}{s} \in \dom{h}
        \\[8pt]
\Run{L}{\TRACE{e}}{s}{h}
& \xrightarrow{~\List{\Trace{\sem{e}{s}}}~}
    \Stop{L}{s}{h}
\end{align*}
\caption{Semantic rules for program execution.}
\label{fig:sem}
\end{figure}

\begin{deleted}

\subsection*{Strong Characterization of Executions}%
\label{sec:property}

[content moved to the appendix as part of the full soundness proof,
unfortunately there isn't enough room in the main body]
\end{deleted}

\changed{While \cref{lem:secure} gives a precise account
of security from in a semantic, glassbox view,
we still lack a formal characterization from the adversarial perspective.
In this section}{} We define attacker knowledge based on the \emph{observation}
that they can make of a schedule and what can be learned to reduce one's \emph{uncertainty}
about possible initial states resp. secrets.
The first main result, \cref{thm:guarantee}, formalizes the promise made in
\cref{sec:motivation} that any gain in knowledge is linked to an earlier assumption failure.

\begin{definition}[Attacker-visible actions and schedules]
    \label{defn:visible}
For an action~$a$, $\visible{a}$ keeps~$a$ if it is visible to an $\ell$ attacker
and erases it into~$\tau$ otherwise.
This definition is lifted to schedules in the obvious way as $\visible{\sigma}$.
{\normalfont
\begin{align*}
\visible{\tau} &= \tau \qquad &
\visible{\Fst} &= \Fst \\
\visible{\Assumption{s~\rho}} &= \tau
    &
\visible{\Snd} &= \Snd \\
\visible{\Trace{e}} &= \tau
    &
\visible{\Load{p}} &= \Load{p} \\
    &
    &
\visible{\Store{p}} &= \Store{p} \\[4pt]
\visible{\Out{\ell'}{v}} &=
    \rlap{$\Out{\ell'}{v}$ if $\ell \sqsubseteq \ell'$ else $\tau$}
    &
    &&
\end{align*}}
\end{definition}
The concurrent schedule is always visible and so are memory accesses.
An output action is visible only for an attacker who is allowed to observe the respective channel.
In contrast, assumption steps are modeled to \emph{not} be visible
because they do not constitute actual observations,
albeit an attacker with knowledge about the program's source code
knows their occurrence and the assumed formula~$\rho$,
due to the fact that executions are always matched.
\begin{definition}[Observably equivalent schedules]
    \label{defn:equivalent}
Two schedules $\sigma$ and $\sigma'$ are observably equivalent
for an $\ell$ attacker, written $\sigma \approx_\ell \sigma'$,
if their $\ell$-visibility projection is the same:
\[ \sigma \approx_\ell \sigma' \quad\defeq\quad
    \visible{\sigma} = \visible{\sigma'} \]
\end{definition}
Note that this implies that the length of $\sigma$ and $\sigma'$ is the same.
\changed{and that $\sigma \cong_\ell \sigma'$ implies $\sigma \approx_\ell \sigma'$
but not vice versa.}{}

Information leakage is phrased in the standard \emph{knowledge-based}
style~\citep{Askarov_Sabelfeld_07,Broberg_Sands_09,Askarov_Chong_12,Broberg_vDS_15}.
This style of security property talks about the attacker's knowledge
in order to state that the attacker doesn't learn anything that should not
have been revealed to them. Knowledge is captured in terms of the
attacker's \emph{uncertainty} \addition{about the program's secret data that the
attacker is not supposed to learn}. Specifically, uncertainty is the complement
of knowledge so decreased attacker uncertainty corresponds to an increase
in attacker knowledge.
The following definition captures this intuition.
Together with $\visible{\_}$ it serves as the \emph{formal specification} of the
adversarial capabilities outlined in \cref{sec:threat-model}.
\begin{definition}[Attacker Uncertainty]\label{defn:uncertainty}
For a given initial state $(s,h)$ and schedule~$\sigma$ for command~$c$,
the attacker must accept as explanations all possible initial states $(s',h')$
which can produce an observably equivalent schedule~$\sigma'$:
\begin{align*}
& \uncertainty{\ell}{P}{\sigma}{c}{L}{s}{h} \quad \defeq \\
& \quad \big\{~ (s',h') \mid \exists\ \sigma'\ k'.\ \
    (s,h), (s',h') \vDash  P \star \invs{L}\ \\
& \hspace*{3cm}
\land \Steps{\Run{L}{c}{s'}{h'}}{\sigma'}{k'} \land \schedcteq{\ell}{\sigma}{\sigma'}~\big\}
\end{align*}
\end{definition}

\changed{}{
The property assesses how the attacker's uncertainty changes over time.
Specifically we can use it to compare the attacker's uncertainty
before and after each execution step. Any decrease in uncertainty represents new
information that the attacker learned from that step. The property requires
that this change in knowledge must be bounded by what the attacker is
permitted to learn by that step of execution:
For the policy-agnostic guarantee, each execution step
is allowed to reveal (i.e.\ decrease the attacker's uncertainty about)
only failed $\ASSUME~\rho$ steps,
which in turn correspond to unsatisfied $\Assumption{\rho}$ actions
in the schedule:}

\begin{definition}[Assumption failure]
An assumption failure occurs at position~$n$
with
$n < |\sigma|$ and $n < |\sigma'|$
in a pair of schedules,
if at that point both contain the same assumption~$\rho$,
and that assumption is not satisfied between the associated stores
recorded in the action.
{\normalfont
\begin{align*}
& \mathsf{assumption\text{-}failed}_\ell(n, \sigma, \sigma') \\
&  \quad \defeq \quad \Exists{s\ s'\ \rho}{
        \sigma_n = \Assumption{s~\rho} \text{ and }
        \sigma'_n = \Assumption{s'~\rho} \text{ and }
        s, s' \not\models_{\ell}} \rho
\end{align*}}
\end{definition}


Complementary \changed{}{to uncertainty}, we formalize what the $\ell$-level attacker is allowed to learn from a single execution step following a known execution prefix from initial state $(s,h)$
with schedule~$\sigma$.
This will be defined as a set $\assumedrelease{\ell}{P,\sigma}{c}{L}{s}{h}$
of initial states $(s',h')$ which the attacker is allowed to \emph{exclude} from their uncertainty by observing such an additional step.
\begin{definition}[Release by assumption]\label{defn:release}
\begin{align*}
& \assumedrelease{\ell}{P}{\sigma}{c}{L}{s}{h} \quad \defeq \\
& \quad \big\{~(s',h') \mid \exists\ \sigma'\ k'.\
    (s,h), (s',h') \vDash  P \star \invs{L}\
 \\
&    \hspace{3cm} \land \Steps{\Run{L}{c}{s'}{h'}}{\sigma'}{k'} \land \schedcteq{\ell}{\sigma}{\sigma'} \\
&    \hspace{3cm} \land \exists\ n.\ n < |\sigma| \land n < |\sigma'| \\
&    \hspace{4cm} \land \mathsf{assumption\text{-}failed}(n,\sigma,\sigma') ~\big\}
\end{align*}
\end{definition}
This definition mirrors \cref{defn:uncertainty} except that
only those initial states $(s',h')$ are kept that can lead to a failed assumption.


\begin{theorem}[Policy-agnostic security guarantee]%
    \label{thm:guarantee}
If $\prove{\ell}{P}{c}{Q}$ then for a major run
$\Steps{\Run{L}{c}{s}{h}}{\sigma_1}{k_1}$
the knowledge gain from one additional step $\Step{k_1}{\sigma_2}{k_2}$,
expressed as the difference in uncertainty,
is bounded by the release condition:
\begin{align*}
& \uncertainty{\ell}{P}{\sigma_1}{c}{L}{s}{h} \setminus
  \uncertainty{\ell}{P}{\Concat{\sigma_1}{\sigma_2}}{c}{L}{s}{h} \\
{} \subseteq {} & \assumedrelease{\ell}{P}{\sigma_1}{c}{L}{s}{h}
\end{align*}
\end{theorem}
\begin{deleted}
The connection between
the information leak and the failed assumption
(cf. red arrow in \cref{fig:traces})
is made as follows.
The corresponding formula~$\rho$
is tracked as part of the verification in the precondition~$P$ of
$\prove{\ell}{P}{c}{Q}$.
Later, in proof rules like \textsc{Output} in \cref{fig:ioa},
when we actually have to \emph{prove} that a value is public,
any information leak that is due to a prior assumption failure
will cause a contradiction with $\rho$ as part of the application-specific reasoning
and thus vacuously entail the consequences of \cref{lem:secure}.
\end{deleted}

\section{Conformance with Policies}%
\label{sec:audit}

With $\mathsf{assumed\text{-}release}_\ell(\_)$
we have a precise characterization of possible information leaks
due to failed assumptions. Next we show how these leaks can be justified
and formally bounded in terms of high-level declassification policies:
\begin{definition}[Declassification policy]
A declassification policy
$\D(\tr) = \policy{\Dpre(\tr)}{\Dpost(\tr)}$
specifies a condition~$\Dpre$ that states \emph{when}
the policy applies and a relational release formula~$\Dpost$
that encodes \emph{what} information is allowed to be released then.
\end{definition}
Both constituents \changed{}{$\Dpre$ and $\Dpost$} may mention the current trace $\tr \colon \mathit{List}\langle\mathit{Event}\rangle$,
a logical list of application-specific $\mathit{Event}$s.
In the example from \cref{sec:motivation}, events are just
the numbers returned from and added to the trace by \avggetinput.

It is sometimes convenient (cf. \cref{sec:case-studies})
to let the formulas range over common auxiliary parameters~$\vec{x}$,
where
$\policy{\Dpre(\tr,\vec x)}{\Dpost(\tr,\vec x)}$,
abbreviates the slightly involved policy
$\policy{(\exists\ \vec x.\ \Dpre(\tr,\vec x))}
        {(\forall\ \vec x.\ \Dpre(\tr,\vec x) \implies \Dpost(\tr,\vec x))}$.
The intuitive reading is just that the condition may bind
some values that are later referred to by the release
which encodes a policy that
declassifies different information depending on a number of cases in $\Dpre$.

In order to track the trace~$\tr$ throughout the verification,
we extend the assertions by a designated abstract history predicate~$\History(\_)$,
and for the sake of presentation we also introduce
an explicit mechanism to extend this trace by an additional command~$\TRACE{e}$
(\changed{recall that }{}in \Tool this is instead realized as library annotations),
where expressions~$\tr$ and~$e$ denote a trace resp. event,
\begin{align*}
\text{assertion } P \ & ::= \
         \cdots \mid \History(\tr)
    &
\text{command } c \ & ::= \
         \cdots \mid \TRACE{e}
\end{align*}
\changed{}{where $\History(\tr)$ says that current value of
expression $\tr$ is the trace until now,
and $\TRACE{e}$ is a specification command that extends this trace by an additional event, denoted by expression~$e$.}

The purpose of an audit of a given verification with respect to a policy
is to inspect each $\ASSUME \rho$ statement placed in the program.
To that end, we need to refer to the \emph{verification context}
at that point, specifically the assertion/path condition~$P$
and trace~$\tr$ that occurs in the sub-derivation
$\prove{\ell}{P \star \History(\tr)}{\ASSUME \rho; \ldots}{\ldots}$
of that program part.
A policy $\D$ is honored if $P$ implies~$\Dpre$
and $\Dpost$ with~$P$ implies~$\rho$ at every such occurrence of assumptions.

In the same spirit as the rest of the paper, we show how the concern
of policy adherence can be separated out of the verification
of the program implementation (cf. comments at the end of \cref{sec:motivation})
in such a way that we can still draw a connection between all respective
constituents.
We define \emph{extended} judgements
\[ \audit{\ell}{P \star \History(\tr)}{c}{Q \star \History(\tr')}{A} \]
that now mention explicitly the history predicate
which is threaded through the proof alongside
all other assertions~\citep{ErnstKnappMurray2022,Schoepe_MS_20,blom2015history,penninckx2019specifying}.
Moreover, we instrument the \changed{derivations}{proof rules} to produce a set~$A$
of \emph{audit triples}~$(P,\tr,\rho)$ from each occurrence of an assumption
as the verification context mentioned above.

\begin{figure}[t]
    \begin{mathpar}
\infer
    {\ }
    {\audit{\ell}{\History(\tr)}{\TRACE{e}}{\History(\tr \cdot \List{e})}{\varnothing}}
    \textsc{Emit}

\infer
    {\ }
    {\audit{\ell}{P \star \History(\tr)}{\ASSUME \rho}{P \star \rho \star \History(\tr)}{\{(P,\tr,\rho)\}}}
    \textsc{Assume}

\infer
    {\audit{\ell}{P}{c_1}{Q}{A_1} \and
     \audit{\ell}{Q}{c_2}{R}{A_2}}
    {\audit{\ell}{P}{c_1; c_2}{R}{A_1 \cup A_2}}
     \textsc{Seq}

\infer
    {\audit{\ell}{P}{c}{Q}{A}
        \and \mathsf{mod}(c) \cap \fv{F} = \varnothing}
    {\audit{\ell}{P \star F}{c}{Q \star F}{\{(P \star F, \tr, \rho) \mid (P,\tr,\rho) \in A\}}}
     \textsc{Frame}
    \end{mathpar}
    \caption{Proof rules for event histories 
             and audit triples.}
    \label{fig:audit}
\end{figure}

Some interesting proof rules are shown in \cref{fig:audit}:
Emitting a trace event~$e$ symbolically extends the trace expression bound by history predicate~$\History(\_)$
to $\Concat{\tr}{\List{e}}$.
Assumptions \changed{record}{produce an audit triple that records} the current proof context $P,\tr$ alongside the assumed formula~$\rho$.
As an example for \changed{structural}{syntax-directed} rules, sequential composition merges
the results from both commands.
Rule \textsc{Frame} shows that any frame condition~$F$ that is preserved by the execution of~$c$ can be adjoined to the audit triples after the fact,
such that alternatively rule \textsc{Assume} could have been formulated as a ``small axiom'' with $\Emp$ instead of a general~$P$, i.e., framing is compatible with recording proof contexts.
\changed{leading to the notion of a \emph{policy audit}:}{Now we can formalize policy audit.}
\begin{definition}[Policy audit]
    \label{defn:audit}
A verification
$\audit{\ell}{P \star \History(\tr)}{c}{Q \star \History(\tr')}{A}$
is correctly audited wrt. a policy 
$\D(\tr) = \policy{\Dpre(\tr)}{\Dpost(\tr)}$
if for each $(P,\tr,\rho) \in A$
implications
$P \implies \Dpre(\tr)$
and $P \star \Dpost(\tr) \implies \rho$
are valid.
\end{definition}

\changed{Intuitively, audit triples are just gaps in the proof.
The conditions are therefore hardly surprising but
like
the policy-agnostic security guarantee we want to
ensure that the intuition is actually semantically justified.}
{Intuitively, audit triples are simply proof obligations
for every assumption to be justified by the declassification policy.
As with the policy-agnostic security guarantee,
we now provide a semantic guarantee that bridges between~$A$ from the calculus
and information release by policy.}
\newcommand{\trace}[1]{\mathsf{trace}(#1)}
We denote by $\trace{\sigma}$ the sequence of values~$e$ from
$\Trace{e}$ actions in the schedule~$\sigma$, defined as
$\trace{\sigma} \mathrel{\defeq} \List{e \mid \Trace{e} \in \sigma}$
with the intention that $\trace{\sigma}$ coincides with the evaluation of
the trace expression~$\tr$ in any post-state that asserts~$\History(tr)$.
Moreover, as $\Dpre(\_)$ can be regarded a formula of one variable,
say $\mathtt{tr}$,
we write $\sigma \models \Dpre$
when $\sem{\Dpre(\mathtt{tr})}{s}$ is true for state $s(\mathtt{tr}) = \trace{\sigma}$
(all other variables in~$s$ are irrelevant),
similarly, we write $\sigma,\sigma' \models \Dpost(\mathtt{tr})$
for $s,s' \models \Dpost(\mathtt{tr})$ and $s(\mathtt{tr}) = \trace{\sigma}$, $s'(\mathtt{tr}) = \trace{\sigma'}$.
We define counterparts to $\mathsf{assumption\text{-}failed}$
and $\mathsf{assumed\text{-}release}$ with respect to policies.
\begin{definition}[Policy exclusion]
    \label{defn:policy-excludes}
A declassification policy~$\D = \policy{\Dpre}{\Dpost}$
excludes a pair of schedules
(from the obligation to prove absence of leaks),
if after some number of steps~$n$ with $n < |\sigma|$ and $n < |\sigma'|$
the declassification condition is satisfied
but the release formula is not:
{\normalfont
\begin{align*}
& \mathsf{policy\text{-}excludes}_\ell(\policy{\Dpre}{\Dpost}, n, \sigma, \sigma')\quad \defeq \\
&  \quad \sigma\cut{n} \models \Dpre \text{ and } \sigma'\cut{n} \models \Dpre \text{ and }  \sigma\cut{n},{\sigma'}\!\!\cut{n} \not\models \Dpost
\end{align*}}
\end{definition}

\begin{definition}[Release by policy]
The initial states $(s',h')$ that an attacker
may remove from their uncertainty
by observing any further step after a prefix run with schedule~$\sigma$ 
are those minor runs that are excluded by policy~$\D$.
{\normalfont
\begin{align*}
& \policyrelease{\ell}{\D}{P}{\sigma}{c}{L}{s}{h} \ \defeq \\
& \quad \big\{~(s',h') \mid \exists\ \sigma'\ k'.\
\Steps{\Run{L}{c}{s'}{h'}}{\sigma'}{k'} \land \schedcteq{\ell}{\sigma}{\sigma'} \\
&   \hspace*{10mm} \land \exists\ n.\ n < |\sigma| \land n < |\sigma'| \land \mathsf{policy\text{-}excludes}(\D,n,\sigma,\sigma')~\big\}
\end{align*}}
\end{definition}

Finally, we can state the second main result:
\begin{theorem}[Policy-specific security guarantee]
    \label{thm:audit}
For a verified program
$\audit{\ell}{P \star \History(\List{})}{c}{Q \star \History(\tr')}{A}$
and a policy~$D$ formally audited according to \cref{defn:audit},
for each major run
$\Steps{\Run{L}{c}{s}{h}}{\sigma_1}{k_1}$
with final step
$\Step{k_1}{\sigma_2}{k_2}$ we have:
\begin{align*}
& \assumedrelease{\ell}{P}{\sigma_1}{c}{L}{s}{h} \\
{} \subseteq {} & \policyrelease{\ell}{\D}{P}{\sigma_1}{c}{L}{s}{h}
\end{align*}
\end{theorem}

Under the conditions of \cref{thm:audit}, 
we get 
(owing to \cref{thm:guarantee})
%
the ultimate property
that every knowledge increase is within the policy:
\begin{align*}
& \uncertainty{\ell}{P}{\sigma_1}{c}{L}{s}{h} \setminus
 \uncertainty{\ell}{P}{\Concat{\sigma_1}{\sigma_2}}{c}{L}{s}{h}  \\
 {} \subseteq {} & \policyrelease{\ell}{\D}{P}{\sigma_1}{c}{L}{s}{h}
\end{align*}

\section{Case Studies}\label{sec:examples}\label{sec:case-studies}

\changed{}{
We demonstrate the approach of this paper with several challenging case studies that we implemented and verified using auto-active verifier \Tool, an extension of \SecC
that adds constant-time security checks for memory access,
adapts the semantics of value classification as described in \cref{sec:logic},
and adds an \texttt{-audit} flag which shows the audit conditions from \cref{sec:audit} to the user. (In all case studies we inline the checks of these
audit obligations
into the verification, as explained \ifExtended at the end of \cref{app:motivation-proof}, \fi
so that \Tool discharges them automatically.) 
\Tool is implemented in the Scala programming language and encompasses roughly 5\,KLoC. Like \SecC, \Tool mechanises the application of our extended \SecCSL logic (i.e.,
the application of the rules in \cref{fig:rules} and \cref{fig:ioa}) by symbolic execution.
The case-studies are written in C with logical definitions and program annotations formulated in the specification language of \SecC.
\Tool inherits some limitations from \SecC:
It supports a significant fragment of C but lacks for example union types,
taking pointers to local variables, and floating point and bit-wise operations.
Numeric types are treated as unbounded mathematical integers (as is common
in auto-active verifiers)
so that the verification is not sound in the presence of overflows.
The absence of overflow can be proved in \Tool by adding
assertions on integer operations.
}

\changed{}{
  The formal soundness theorems \cref{thm:audit} and \cref{thm:guarantee},
  mechanised in Isabelle/HOL,
apply to the simple command language
formalised in this paper (whose semantics is given in \ifExtended\cref{fig:sem-cmd}\fi). So \Tool's
soundness follows from those theorems, so long as \Tool correctly implements
the semantics of its subset of C and correctly implements the rules of
\cref{fig:rules} and \cref{fig:ioa}.
Aside from treating \texttt{int}s as unbounded integers we believe \Tool is
faithful to the formal program semantics and correctly implements the logic.
Both of these assumptions could in
principle be discharged by applying
orthogonal ideas on validating the output of auto-active verifiers~\citep{parthasarathy2021formally,jacobs2015featherweight}.
}

\begin{table}
  \begin{added}
\begin{tabular}{ p{2.2cm} p{1cm} p{1cm} p{1.5cm} p{1cm}   }
\toprule
{\bf Case Study} & {\bf Proof Ratio} & {\bf Verified SLOC} & {\bf Unverified SLOC} & {\bf Effort (pw)}\\
\cmidrule(r){1-1}
\cmidrule(lr){2-2}
\cmidrule(lr){3-3}
\cmidrule(lr){4-4}
\cmidrule(l){5-5}

Location Service & 1.9  & 210 & 124 & 3\\

Auction Server & 4.3 & 187 & 79 & 3.5\\

Wordle & 5.9 & 47 & 81 & 0.3\\

Private Learning & 2.5 & 315 & 114 & 6\\
\bottomrule
\end{tabular}
\end{added}
\caption{\changed{}{Case study statistics. We report the size of the case studies in
  Source Lines of Code (SLOC), including the size of the {\bf Verified} code;
  the {\bf Unverified} code; and the {\bf Proof Ratio}, the ratio of the size of
  the proof
  (definitions, lemmas, specifications etc. as \Tool annotations) to that of the verified code.
  We also report the total {\bf Effort} in person-weeks {\bf (pw)}.}\label{tbl:stats}}
  
\end{table}

\subsection{Differentially-Private Location Service}\label{sec:location}

Our first case study implements a multi-threaded,
privacy-preserving \emph{location service}. Such a service might run,
for instance, on  a user's mobile phone. Its intention is to release
information about the user's location, but in accordance with a privacy
policy that implements 
differential privacy~\citep{10.1007/11787006_1} for mobility
traces~\citep{10.1007/978-3-319-08506-7_2} (i.e.\ traces of reported locations
for the user).

\begin{figure}
  \centering\includegraphics[width=0.65\columnwidth]{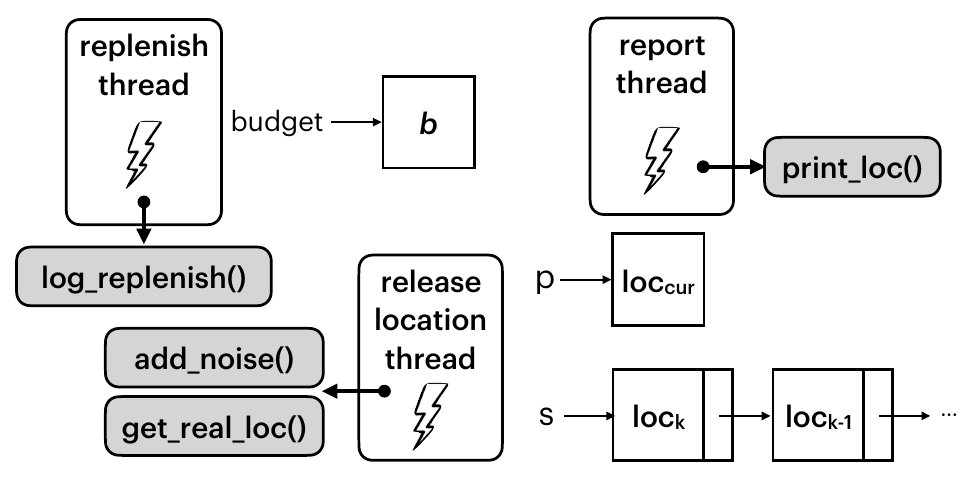}
  \caption{Architecture of the location service.\label{fig:location} External
  functions are coloured grey.}
\end{figure}

This service, whose architecture is depicted in \cref{fig:location}, contains three threads, that all run continuously:
the \emph{release location} thread periodically releases information
about the user's location, in accordance with a differential privacy
policy. It copies information about the user's most recent physical
location, obtained
via the external function \getrealloc, into
the heap location pointed to by \varp, after adding noise to ensure differential
privacy, using the external function \addnoise. This function is an
off-the-shelf implementation of the \emph{planar laplacian}~\citep{Geo_2013}
geolocation privacy mechanism. 
The state of the differential privacy
policy is recorded in the location pointed to by \budget. Finally, this
thread maintains a (linked) list~\vars carrying information about the user's
prior locations.

This linked list~\vars is used in situations where the privacy
budget has been exhausted. In particular, each time noise is added to the
user's real, private location via the external \addnoise function, some of
the user's privacy budget is consumed.
When the budget gets sufficiently low, no new information about the user's
location can be revealed without violating differential privacy. In this
situation, the \emph{release location} thread instead applies extrapolation
to \emph{predict}~\citep{10.1007/978-3-319-08506-7_2}
the user's most likely current location from the location
history recorded in the linked list \vars. Importantly, this history
contains only values
already previously released, i.e.\ noisy values resulting from previous
applications of the differential privacy mechanism \addnoise. Hence,
releasing a location prediction made from those values reveals no new
information, while still allowing the \emph{release location} thread to
provide continuous service. 

The \emph{report thread} periodically reports the user's most recent location
as recorded by the \emph{release location} thread in~\varp.
Finally, the \emph{replenish budget} thread periodically replenishes the
privacy budget \budget, allowing the \emph{release location} thread to again
apply the differential privacy mechanism to release new (noisy) location points
(rather than predictions). A global lock controls access to all shared data structures:
$\varp$, $\vars$ and $\budget$.

\myparagraph{Security Policy.}
The policy ensures that the
differential privacy mechanism is appropriately applied and that no
raw location data (to which the mechanism has not been applied) is ever
revealed. Specifically, the external function \printloc has a contract
that requires its argument is $\low$; yet the contract for
\getrealloc says that the returned location is $\high$
so that in
the absence of declassification, no location data can ever be revealed by
the service.

The declassification policy allows releasing location data only via the
correct use of the
differential privacy mechanism. Traces~$\tr$ in this example are sequences of events, each of which is either an
event~$\Consumed{\noiselat}{\noiselon}$, recording that the differential
privacy mechanism produced a noisy location point
$(\noiselat,\noiselon)$ and consuming a fixed positive amount~$\epsilon$
of the privacy budget; or the event $\Replenished$ recording that the
privacy budget was replenished by a fixed positive amount~$r$.
Both $\epsilon$ and~$r$ are public ($\low$)
constants that control
the strength of the privacy guarantee.

The contract for the external function $\addnoise$ that implements the
differential privacy mechanism is as follows. For brevity we elide,
via $\ldots$ in the precondition, that the argument~\texttt{pt} is a valid
pointer to a \texttt{struct point}, and freely intermix ASCII and
mathematical notation.
\begin{lstlisting}
  void add_noise(struct point *pt);
  _(requires %$\exists \tr$%. %$\History(\tr)$% ...)
  _(ensures  %$\exists \noiselat\ \noiselon.\ \History(\Concat{\tr}{\Consumed{\noiselat}{\noiselon}})$%)
  _(ensures  &pt->lat %$\mapsto \noiselat\  \star $% &pt->lon %$\mapsto \noiselon$%)
\end{lstlisting}
Whenever this function is called an appropriate event is
recorded in the trace to remember that some privacy budget was consumed.
%
Similarly, $\Replenished$ events are generated by the external function
\logreplenish, which is called by the \emph{replenish} thread when
replenishing the budget. 

The declassification policy~$\D_{\mathsf{loc}}(\tr,e)  = \varphi_{\mathsf{loc}}(\tr,e) \leadsto \rho_{\mathsf{loc}}(tr,e)$  then specifies that it is safe to declassify the value~$e$
only when~$e$ is the most recent point~$(\noiselat,\noiselon)$ generated by the
differential privacy mechanism \addnoise only when there is sufficient
privacy budget.
\[
\begin{array}{l}
\varphi_{\mathsf{loc}}(\tr,(\noiselat,\noiselon))\  \defeq\ 
\History(\tr) \ \star  \countbudget{\tr} \geq \epsilon \star {} \\ 
\hspace*{3cm} \exists\ tr'.\  tr = \Concat{\tr'}{\Consumed{\noiselat}{\noiselon}} \\\
\rho_{\mathsf{loc}}(\tr,(\noiselat,\noiselon))\ \defeq\  (\noiselat,\noiselon) :: \low
\end{array}
\]
The function $\countbudget$ (defined by straightforward recursion on $\tr$)
iterates through $\tr$ to count up the budget remaining
at the present time---consume events add~$-\epsilon$ while replenish events
add~$r$ to the budget, which starts at 0.  

\changed{}{This security policy demonstrates how our approach securely handles
  multiple (or repeated) declassifications in the same execution of a
  single program: each act of declassification is justified separately by
  appealing to the policy wrt. the trace $\tr$ at that point in time.}

\changed{}{It is important to note that this security policy does not
  verify that the differential privacy mechanism has been appropriately
  \emph{implemented}; instead it verifies that it is appropriately
  \emph{used} by the application. In other words, this policy does not
  directly state a differential privacy guarantee; however it ensures such a
  guarantee  under the assumption that the differential privacy mechanism
  correctly provides the privacy guarantee represented by the
  deplenishing budget. Verifying the mechanism implementation of course
  requires probabilistic reasoning and is best done using other approaches~\citep{barthe2016proving}. Similar arguments apply to the private learning
case study later on (\cref{sec:maclearning}).}

\myparagraph{Verification.}
Since \Tool does not currently support reasoning about floating
point arithmetic, for this case study such arithmetic (e.g.\ in the verified
extrapolation code) is modelled using integer arithmetic instead. This is sound
since the program invariant and security policy do not depend on any floating-point
arithmetic.
We prove the verified extrapolation code operates only on public values,
ruling out
timing channels from floating point operations~\citep{kohlbrenner2017effectiveness}. 
The verified part of this application comprises 210 source
lines of code (SLOC). Unverified code (124 SLOC) comprises 
the
external functions that allocate
memory, acquire and release the global lock, implement the planar laplacian
mechanism (72 SLOC), and simulate generating user location points, as well as the
\texttt{main()} function that sets up the threads.
The proof:code ratio (proof SLOC to verified code SLOC) is~1.9.
This effort to complete this case study was approximately 3~person-weeks \changed{}{(see \cref{tbl:stats})}.

\subsection{Sealed-Bid Auction Server}\label{sec:auction}

Our second case study is a \emph{sealed-bid} auction server. In such an
auction, all bids are kept secret until the auction is complete. This prevents
bidders racing to outbid one another. Our server
uses concurrency to ensure that no bidder can deny service to another
by servicing each client connection
in a separate thread.

Each client bid is handled
by a separate \emph{handle bid} thread. A separate \emph{close auction} thread
waits until the auction duration (a fixed, public parameter) has elapsed and
then closes the auction. Both make use of external logging functions: the
\emph{handle bid} thread logs incoming bids using the \logbid external function,
while the \emph{close auction} thread logs the fact that the auction has closed
using the \logclosed function. Logging is important in this case study 
to provide an audit trail (e.g.\ in case of a disputed auction). 
The \emph{close auction} thread uses the external \printresult function
to print out the result of the auction, once it has been closed.

When receiving a new bid, the \emph{handle bid} thread compares the
newly submitted bid to the current maximum in a constant-time fashion,
and updates the latter if the new bid is larger.
Bids are pairs $(\bidid,\bidqt)$ where $\bidid$ is the identity of the client
who submitted the bid and $\bidqt$ is the amount (or \emph{quote}) of the bid.

\myparagraph{Security Policy.}
The top-level verified function that implements the
\emph{handle bid} thread takes as its argument the bid to be handled. The
precondition on this function states that the bid is secret ($\high$). Thus
all bids are treated as secret. The
precondition on the external \printresult function requires that its argument
is public ($\low$). So the only way for a winner to be announced is via
declassification. \Tool's constant-time guarantee meanwhile ensures that no
information about bids can be leaked (including via timing channels) prior to
declassification.

The declassification policy states that no bid information can be declassified
until after the auction is closed. At this time the only bid that
can be declassified is the maximum bid that was received (i.e.\ the auction
winner). Hence, the policy allows only the winning bid to be revealed only
after the auction has closed, when the attacker learns the winning
bid and that no other bid was higher.

To capture this policy, the trace~$\tr$ records two kinds of events:
$\Running{\bidid}{\bidqt}$ represents the submission of a bid from the client
with id $\bidid$ and amount $\bidqt$ while the auction is running;
$\Finished$ represents that the auction has been closed.
Each is generated by one of the external logging functions: $\Running{\bidid}{\bidqt}$ is generated by the \logbid function, while $\Finished$ is generated by
\logclosed. In this way we piggy back on the application's normal functioning to
define its security policy.
The contracts for these external functions are similar to those for
\addnoise and \logreplenish from \cref{sec:location}.

The declarative, extensional declassification 
policy \linebreak $\D_{\mathsf{bid}}(\tr,e)\ \defeq\ \varphi_{\mathsf{bid}}(\tr,e) \leadsto \rho_{\mathsf{bid}}(\tr,e)$ is defined
as follows, where~$e$ is the value~$(\bidid,\bidqt)$ to be declassified.
\[
\begin{array}{l}
  \varphi_{\mathsf{bid}}(\tr,(\bidid,\bidqt))\ \defeq\ \History(\tr)\  \star \exists \tr'.\ \tr = \Concat{\tr'}{\Finished} \star {} \\
  \hspace*{5mm}\contains(\tr',\Running{\bidid}{\bidqt}) \star \ismax(\tr',\bidqt) \smallskip \\ 
  \rho_{\mathsf{bid}}(\tr,(\bidid,\bidqt))\ \defeq\ (\bidid,\bidqt) :: \low
\end{array}
\]
where $\contains(xs,y)$ is the standard list function for testing whether list~$xs$
contains the element~$y$, and $\ismax(\tr',\bidqt)$ checks that
$\tr'$ contains only events $\Running{\bidid'}{\bidqt'}$ for which
$\bidqt \geq \bidqt'$ and is defined via recursion on $\tr'$.

\myparagraph{Policy Composition.}
To evaluate our approach's ability to handle the \emph{composition} of multiple
security policies, we decided to extend the example and
augment its security policy. Specifically, we added a
\emph{reserve price} feature to the auction server. When run in this mode, the
user supplies a (secret)
reserve price and in order for a winner to be declared, there
must be a bid that is greater-or-equal to this reserve. We parameterise our
verification by an arbitrary, constant reserve price~$r$  and use the abstract
separation logic predicate $\Reserve(r)$ to denote that the auction is running
in the reserve price mode and that $r$ is the reserve
price.

In this mode, all bidders learn whether any bid was $\geq$ the reserve
since, if it was not, no winner is announced. Therefore the 
policy allows this (boolean) fact~$\met$ to be declassified unconditionally (only) once the
auction is closed. 
\[
\begin{array}{l}
  \varphi_{\mathsf{met}}(\tr,\met)\ \defeq\ \History(\tr) \star \exists\ \tr', r.\ \Reserve(r)  \star \tr = \Concat{\tr'}{\Finished} \star {} \\ 
  \hspace*{4cm} \met = \resmet(\tr',r) \smallskip \\[4pt]
  \rho_{\mathsf{met}}(\tr,\met)\ \defeq\ \met :: \low
\end{array}
\]
$\resmet(\tr',r)$ is a simple recursive function that iterates through $\tr'$
returning $\trueconst$ as soon as it finds a bid $(\bidid,\bidqt)$ for which
$\bidqt \geq r$, or $\falseconst$ if none is found.
The declassification policy for the winning bid is then specified as follows,
using $\varphi_{\mathsf{bid}}$ and $\rho_{\mathsf{bid}}$ defined for the non-reserve mode earlier.
\[
\begin{array}{l}
  \varphi_{\mathsf{resbid}}(\tr,(\bidid,\bidqt))\ \defeq\ \History(\tr) \star \Reserve(r) \ \star \bidqt \geq r \star {} \\ 
  \hspace*{4cm} \varphi_{\mathsf{bid}}(\tr,(\bidid,\bidqt)) \smallskip \\[4pt]
  \rho_{\mathsf{resbid}}(\tr,(\bidid,\bidqt))\ \defeq\ \rho_{\mathsf{bid}}(\tr,(\bidid,\bidqt))
\end{array}
\]
The \changed{}{composed} declassification policy is of course the non-overlapping disjunction of
\addition{the mutually-exclusive predicates}
$\D_{\mathsf{met}}$ and $\D_{\mathsf{resbid}}$, depending on the type of
the second argument.

\myparagraph{Verification.}\label{sec:auction-verification}
%
%
The total size of the verified code for this application is 187 SLOC.
Unverified code (79 SLOC) comprises
external functions that allocate
memory, acquire and release the global lock, implement logging and printing,
as well as the code that reads from client socket connections, and the
\texttt{main()} function that sets up the threads and the TCP listen socket.
The size of the verified artifact for this case study is
985 source lines, making the proof:code
ratio 4.3. This ratio is higher than previous 
because this case study
involves a lot of meta-level reasoning by induction about the policy
itself. 
%
This case-study required approximately 3 person-weeks of effort; adding
the reserve price feature added 4 person-days \changed{}{(see~\cref{tbl:stats})}.

\subsection{Wordle}
\label{sec:wordle}

Our third case study is a constant-time implementation of the popular game Wordle.
The implementation is a simple server that allows players to connect and to
guess a pre-chosen 5-letter
word~$w$. In response to a player submitting her guess~$g$, the
server replies with a 5-byte response~$r$. The $i$th byte~$r_i$
of the response
provides information about the $i$th letter~$g_i$ of the guess in relation to
the pre-chosen word~$w$: 0 (black) indicates that $g_i$ is not present anywhere
in~$w$; 1 (yellow) that it is present in~$w$ at some index~$j$ for which $g_j \not= w_j$; 2 (green) that $g_i = w_i$.
The server runs a separate thread to service each client connection.

\myparagraph{Security Policy.}
The security policy says that the pre-chosen word~$w$
is secret: $w :: \high$. Each player~$p$ is assigned a distinct
security level~$\ell_p$. A  player's guess~$g$ is known only to herself:
$g :: \ell_p$, encoded in the postcondition of the external
library function that retrieves the player's guess. 
%
The external function that transmits the server's response
back to the player requires in its precondition
that the response~$r$ is allowed to be known to the player: $r :: \ell_p$.
Since~$r$ is a function of~$w$, this requires declassification.

The policy condition~$\varphi_{\mathsf{word}}(\tr)$ for the declassification policy
therefore requires that the player has submitted a most recent
guess~$g$, for the
pre-chosen word~$w$, which is encoded in the trace by appropriate
events that record together for each submitted guess~$g$ the player~$p$ who submitted
it as well as the guess~$g$ itself. An abstract separation logic predicate
is used to remember which is the pre-chosen word~$w$ (similar to the
$\Reserve$ predicate of \cref{sec:auction})

The policy release formula~$\rho_{\mathsf{word}}(\tr)$ is more interesting.
It specifies what information player~$p$ (with security level~$\ell_p$)
is allowed to learn after submitting guess~$g$ for the pre-chosen word~$w$,
and requires that this information is not revealed to anyone else.
Its core is the following:
\newcommand{\ccontains}{\mathsf{ccontains}}
\[
\begin{array}{l}
  \forall i. i \haslabel \low \implies  (i \geq 0 \land i < \length(w) \implies (w_i = g_i) :: \ell_p) \ \star  {}\\
  \hspace*{5mm} \ccontains(w,g,g_i,\length(w)) :: (w_i \not= g_i ? \ell_p : \high)
\end{array}
\]
The first conjunct says the player is allowed to learn whether each
letter~$g_i$ of the guess is equal to the corresponding letter of the
word~$w_i$. The second conjunct says that additionally, if $g_i \not= w_i$,
then the player is allowed
to learn whether~$g_i$ is contained elsewhere
in the word at some location~$j$ for which $w_j \not= g_j$. That is the
result returned by the function~$\ccontains(w,g,g_i,\length(w))$ which is defined
by straightforward recursion on the length of the word. Specifically,
$\ccontains(w,g,c,n)$ returns true if $c$ is contained in the first~$n$
characters of~$w$ at some location~$j$ where $w_j \not= c$.

Note that this security guarantee ensures the server does not leak
information in its timing behaviour about the player's guess, which might
otherwise be exploited by other players to draw extra inferences about the
word~$w$ beyond what they could deduce from their guesses alone.

\changed{}{\cref{tbl:stats} reports the size and effort for this case study.
  The significantly higher proof:verified code ratio is because these proofs
  contain a large amount of generic meta-level reasoning (e.g.\ about lists
  and strings, etc.) required for this case study.}

\subsection{Private Learning}
\label{sec:maclearning}
Our fourth case study investigates the application of our ideas to
secure, private learning. We consider a scenario in which a client
wishes to compute a model, possibly in collaboration with others,
over very sensitive data (e.g.\ parental income, race, and gender to 
predict earnings distributions, incarceration rates \citep{chetty2014land, chetty2018opportunity, chetty2019practical}, 
survival prediction of lung cancer patients \citep{deist2020distributed}, etc.). It is common for such models to be computed in hardware-supported
\emph{secure enclaves}~\citep{hunt2018chiron, hynes2018efficient, kunkel2019tensorscone, mo2021ppfl, TB19a}
provided by trusted execution environments like
Intel SGX \citep{10.1145/2487726.2488368} 
and ARM TrustZone \citep{alves2004trustzone}, to defend against data theft
including against the host operating system. Here, \Tool's constant-time
guarantee is especially relevant, given that TEEs are known to leak data
via various side-channels \citep{10.5555/3154768.3154779, lee2017inferring, moghimi2017cachezoom, moghimi2020copycat, puddu2020frontal};
 enforcing constant-time ensures side-channels cannot be exploited.

 In this case study, a client is invoked with initial model parameters
 $\theta^0$. It runs $\mathsf{T}$ training iterations.
 At each iteration~$t$ ($1 \leq t \leq \mathsf{T}$) it
 refines the model parameters, producing new ones~$\theta^{t+1}$. It does
 so by applying differentially-private gradient descent~\citep{abadi2016deep} (DP-GD), in which a noisy gradient against the model's
 loss function is computed and
 then used to refine the model parameters. The goal is to ensure that the
 refined model~$\theta^{\mathsf{T}+1}$ does not leak too much information
 about the sensitive training data. For simplicity,
 our current implementation
 learns a linear model over the training data. We refer to the process
 in which $\theta^{\mathsf{T}+1}$ is computed from $\theta^0$ as a
 training \emph{epoch}, comprising $\mathsf{T}$ training iterations.
 
 Each training iteration consumes~$\epsilon$ privacy budget; by composition,
 each epoch consumes $\mathsf{T}\cdot\epsilon$. As in the
 location service case study (\cref{sec:location}), a separate thread may
 periodically replenish that budget, if desired. 
 This design allows the client to be deployed in a federated learning
 setup in which clients periodically
 provide their updated model parameters~$\theta^{\mathsf{T}+1}$ to a central server, which then e.g.\ computes the average across all client models,
 before sending that average back to each client to
 use as~$\theta^0$ for a subsequent training epoch.
 Distrusting clients can thus compute a shared model without revealing
 their sensitive data to each other, nor the server.

\myparagraph{Security Policy.}
The initial model parameters~$\theta^0$ are \low, but the
client's training data over which the updated parameters~$\theta^{\mathsf{T}+1}$ are computed
are \high. The updated parameters are
required to be \low, and so must be declassified.

Similarly to the location service (\cref{sec:location}), traces record
events to remember when $\epsilon$ privacy budget is consumed (on each
training iteration) and when the budget is replenished. They also
record events to remember when initial model parameters are
received by the client at the start of each epoch, and when updated
model parameters are released by the client at the end. 
Hence the declarative policy says that the updated
model parameters can be declassified only when they have been correctly
computed ($\mathsf{T}$ training iterations have occurred in the most recent
epoch), for which there
was sufficient privacy budget~$\mathsf{T}\cdot\epsilon$
available before the epoch began.

\myparagraph{Verification.}
The verified part of this case study comprises 315 SLOC, and 114 unverified SLOC
whose functionality is similar to the prior case studies.
The total verified artifact comprises 1108 source lines, yielding a
proofs:code ratio of 2.5. No effort was made to optimise
this ratio and indeed these proofs contain a certain amount of duplicated
lemmas from other case studies.

\section{Conclusion and Related Work}
\label{sec:related}\label{sec:conclusion}

We presented a principled methodology for proving secure
declassification for non-trivial, concurrent, programs.
We decompose the problem into (a)~proving that the program
only leaks information it has explicitly declassified (via \ASSUME statements);
and (b)~auditing the declassifications against a declarative security
policy~$\D$ to ensure that all leaks accord with the policy.
We provide a sound program logic,
supported by the auto-active verifier~\Tool and applied it to
reason about the implementations of various case studies on the order of hundreds of
source lines of code.

In practice, one can of course choose to inline the policy audit (\cref{defn:audit}) into the verification \changed{}{(this is illustrated \ifExtended at the end of
\cref{app:motivation-proof}\fi)},
or alternatively represent the declassification step that appeals
to a policy by a specification-only procedure
with precondition $\Dpre$ and postcondition~$\Dpost$;
or alternatively to place the respective audit conditions into the code.
By disentangling contributions~(a) and~(b) in our formal development
we contribute a justification for this kind of reasoning with respect to a
semantic characterization of attacker knowledge~\citep{chudnov2018assuming}.
Similarly, our ideas are not necessarily tied to the presentation as an extension of the specific foundation \SecCSL.
With the appropriate care to semantic variations (e.g. timing-sensitivity),
we think it is feasible to adapt the approach to other foundations like
modular product programs~\citep{Eilers2018} as implemented in Viper.

Prior work on practical secure declassification includes
the verification of the
kernel of a conference management system~\citep{popescu2021cocon}, a
social media platform~\citep{bauereiss2018cosmed} and its distributed successor~\citep{Bauereiss_GPR_17}.
These works proved variants of the generic security property of
Bounded Deducibility~\citep{popescu2021bounded}, which is similar to
declassification policies~$\D$. The
proofs use manual unwinding in Isabelle/HOL,
over an abstract program representation of I/O automata. 
Li et al.~\citep{Li_LGNH_21}
verified secure declassification policies while 
verifying a 3.8K SLOC Linux KVM hypervisor, in the proof assistant Coq.
Their policies were encoded non-declaratively by
artificially modifying the semantic model
to replace declassified sensitive data with
non-sensitive data, allowing declassification to be proved in
terms of standard noninterference. 

\citet{banerjee2008expressive} enhance the knowledge-based security property of~\citet{Askarov_Sabelfeld_07} with relational assumptions (not using that term), and propose enforcement using a security type system together with relational verification of the declassifying code.  Their declassification policies combine the assumption with an assertion, which should refer to ghost state modeling external observations.  Their formalization is for deterministic sequential programs and does not include the requisite relational logic.
\changed{}{We show the approach can be applied to concurrent programs as well.
We decouple meaning of assume and meaning of policies  (cf.~their Def 5.5), such that assume statements have meaning independently of a stated policy.
Our proof system (\cref{sec:logic}) and audits (\cref{sec:audit}) provide a way to formally establish the requirements outlined by their Definition 6.2 points 2 and 3.
Our explicit trace predicate $\History(\tr)$ realizes their suggested ghost state.
By contrast with their suggestion to use a type system for some relational reasoning, we use 
only the proof system, which encompasses relational reasoning.
}



\changed{}{
\citet{BalliuDG11} observe that knowledge-based properties like these are closely related to standard semantics of epistemic logic, and show how several properties from the literature can be expressed in epistemic temporal logic (but this work does not address verification of such properties, nor concurrent programs).}

\changed{}{\citet{AskarovCM15} formulate knowledge-based security for monitoring of concurrent programs with synchronization in the form of barriers; their monitor is hybrid in the sense that it relies on an oracle for static analysis of branches not taken.  
Compared with a logic or static analysis,
monitoring has the advantage that it can allow use of a program under conditions when its execution is secure, even if the program is not secure in general. 
Monitoring has the disadvantage of significant runtime overhead.
Owing to nuanced use of rely-guarantee reasoning and annotations that designates assumptions a thread makes about locality of shared variables
(adapted from~\citet{Mantel_SS_11}), their monitor is factored into local and global parts and avoids the need for additional synchronization.
}

\changed{}{There is an extensive literature on verification of constant-time security properties;
a recent example is \citet{ShivakumarBGLP22} which also addresses the role of compilers 
in mitigation.  Language based mitigations of timing channels have been studied since \citet{Russo08}.
}

\changed{}{There is also an extensive literature on information flow for concurrent programs.
Prior to the emergence of knowledge-based formulations many variations were based on specialized bisimulations (e.g.,~\citet{SabelfeldS00}).  There are tradeoffs between permissiveness and compositionality of the different properties
(see e.g.,~\citet{Mantel_SS_11}),
and differing models are of interest depending on adversary models.
Surprisingly the property of \citet{SabelfeldS00} is decidable~\citep{Dam06} provided the data model is sufficiently simple for the expression language to be decidable.
}


The \Veronica logic proves secure declassification for shared-memory concurrent
programs~\citep{Schoepe_MS_20}.
\changed{\Veronica supports some
reasoning about programs that might branch on secrets, whereas we prove
a constant-time guarantee which forbids it. Inspired by \Veronica, our
security policies~$\D$ are also defined over external program traces.}{
Its security property is also formulated
as a knowledge-based one. It is more permissive than our constant-time
property in that it can tolerate some secret-dependent branches. However,
to avoid occlusion anomalies~\citep{Sabelfeld_Sands_09},
such branching is disallowed for secrets involved in declassification.
}

\Veronica supports only \emph{unary} (non-relational) assertions,
the entire logic is designed around and fundamentally tied to this principle.
Lacking relational assertions like $e \haslabel \ell$, however,
limits expressiveness and precludes scalability. 
Instead of writing $e \haslabel \low$ for example,
in \Veronica one has to precisely specify the value of expression~$e$
and where it was sourced from,
e.g. $e = x + 5 \land x = \mathit{Low\_Inputs}[3]$ would say that $e$ is the sum of the third input obtained from a low source and the constant five.
Writing invariants (the hard part of verification) in this style quickly becomes impractical, notably for advanced concepts like pointer structures.
For that reason \Veronica is not adequate for programs over arrays or pointers
(none of their examples uses them).
Lack of relational assertions also
means that \Veronica cannot encode declassification policies
like that of the Wordle case study.
From a more practical perspective, \Veronica is not implemented
in a dedicated auto-active verifier like our tool, \Tool.
Overall, our case studies from \cref{sec:case-studies}
are far beyond the scope what can reasonably be verified in \Veronica
and this assessment has been confirmed by \citet{Schoepe_MS_20} in personal communication.


Smith enforces a form of secure declassification
called Qualified Release via so-called \emph{declassification predicates}~\citep{Smith_22}.
Declassification occurs via dedicated
\textbf{declassify} statements, annotated by
unary predicates~$P(e)$ over the value~$e$ to be declassified that are
evaluated in the program's \emph{initial} state. This
notion is soundly enforced in a security type system,
encoded in the auto-active verifier Dafny, and applied to
programs of a few SLOC each against simple policies. 

Our work highlights the difficulty of proving strong constant-time guarantees for
intentionally-leaky
application code: such reasoning necessarily treats
implementation concerns and so cannot be performed
on an abstract model alone. It would be interesting therefore to extend existing constant-time
programming languages~\citep{cauligi2019fact} with support for rich security policies and declassification.

\begin{acks}
We sincerely thank the anonymous reviewers for their comments and insightful suggestions
that enabled us to improve this paper.

This research was sponsored by the \grantsponsor{GSONR}{U.S. Department of the Navy, Office of Naval
Research}{https://www.nre.navy.mil/}, under award \grantnum{GSONR}{N62909-18-1-2049}. Any opinions, findings, and conclusions or recommendations expressed in this material are those of the author(s)
and do not necessarily reflect the views of the Office of Naval Research.

This material is based upon work supported by the \grantsponsor{NGTF}{Commonwealth of Australia Defence Science
and Technology Group, Next Generation Technologies Fund (NGTF)}.

Naumann was supported in part by \grantsponsor{GSNSF}{NSF}{https://www.nsf.gov} award \grantnum{GSNSF}{CNS-1718713}.

\end{acks}

\bibliographystyle{ACM-Reference-Format}

\bibliography{references}

\ifExtended
\clearpage

\appendix

\begin{added}
\section{Proof of the Motivating Example}\label{app:motivation-proof}

Recall (\cref{sec:motivation}) that the example of
\cref{fig:motivation} we verify by defining a resource invariant that
links the input/output history $\tr$ of the program and its
state, to which the pointer \texttt{struct avg\_state * st} points:
\begin{align*}
\inv{\tr,\texttt{st}}
    \quad \defeq \quad
        \History(\tr)    & \land \texttt{st->count} = \length(\tr) \\
                {} & \land \texttt{st->sum} = \sumfunc(\tr)
\end{align*}

Recall that we attach this resource invariant to the lock,
which is implemented by the \texttt{avg\_lock()} and
\texttt{avg\_unlock()} functions that, respectively,
soundly acquire and release this invariant~\citep{Ernst_Murray_19}
by ensuring that it always holds whenever the shared state~\texttt{st}
is accessed.
This we specify via standard annotations as follows
(recalling that \texttt{\bfseries result} in a postcondition
refers to the function's return value):

\begin{lstlisting}
struct avg_state * avg_lock();
_(ensures %$\inv{tr,\texttt{\bfseries result}}$%)

void avg_unlock(struct avg_state *st);
_(requires %$\inv{tr,\mathtt{st}}$%)
\end{lstlisting}

\newcommand{\prooftext}[1]{\textcolor{purple}{#1}}
\newcommand{\proofstate}[1]{\prooftext{\{#1\}}}

Recall that the declassification policy~$\D(\tr)$ for this
example is:
\begin{align*}
\D(\tr) \quad \defeq \quad \policy{\length(\tr) \ge 6~}{~\sumfunc(\tr)/\length(\tr) \haslabel \low}
\end{align*}
The average of the inputs can be declassified so long as there are at least
6 inputs.

Then the proof appears in \cref{fig:motivation-proof}. The proof is carried
out using the rules of our logic (\cref{sec:logic}), automated by \Tool; intermediate proof
states we annotate in \prooftext{purple}, so the reader can see how the
proof progresses. Doing so proves the policy-agnostic
security guarantee: the program leaks no more information than that
contained in assume statements. To prove all leakage is in accordance with the
declassification policy, we apply the rules of \cref{sec:audit} to collect
audit obligations. In this case, we need to prove that the underlined path
condition in \cref{fig:motivation-proof} (what is known at the time of
the assumption) is sufficient to justify the assumption against the
declassification policy. Specifically, letting~$P$ be the underlined path
condition, we have to check (\cref{defn:audit}) that $P$ implies the
policy condition
$\length(\tr) \ge 6$ and that $P$ and the policy release formula
$\sumfunc(\tr)/\length(\tr) \haslabel \low$ together imply the
assumption $\texttt{avg} :: \texttt{\bfseries low}$. These trivially hold,
thus the program satisfies the secure declassification against
this declassification policy: the program leaks no more information than
that allowed by its declassification policy (by \cref{thm:audit}).

In practice, (see \cref{sec:conclusion}) we often inline the check of the
audit obligation (\cref{defn:audit}) into the proof, so it can be automatically
discharged by \Tool. This can be done for instance by replacing
the line  \texttt{\_({\bfseries assume} avg :: {\bfseries low})} in \cref{fig:motivation-proof} with the following three:

\begin{lstlisting}
  _(assert %$\length(\tr) \ge 6$%) // check that %$P \implies \Dpre$%
  _(assume %$\sumfunc(\tr)/\length(\tr) \haslabel \low$%) // assume, therefore %$\Dpost$%
  _(assert avg :: low) // check that %$P \star \Dpost \implies \rho$%
\end{lstlisting}

We explain this transformation for an arbitrary declassification
policy $\policy{\Dpre}{\Dpost}$ whose condition is $\Dpre$ and
release formula is $\Dpost$, and for assumption~\texttt{\_({\bfseries assume} $\rho$)},
as indicated in the comments.
The first line checks that the policy condition~$\Dpre$ holds,
under the current path condition (called $P$ in \cref{fig:motivation-proof}).
This is the first check of \cref{defn:audit}.
Having proved that $\Dpre$ holds, the second line then makes use of the
policy release formula~$\Dpost$. The path condition after the
second line is therefore $P \star \Dpost$. Thus the fourth line
then checks that the original assumption (in this case $\texttt{avg} \haslabel \texttt{\bfseries low}$)
holds, i.e., writing $\rho$ for this assumption, that
$P \star \Dpost \implies \rho$, the second check of \cref{defn:audit}.

\begin{figure*}
  \begin{added}
\begin{lstlisting}
void avg_sum_thread() {
  while(true) {
    struct avg_state * st = avg_lock();
    %$\proofstate{\inv{tr,\mathtt{st}}}$%
    %$\proofstate{\History(\tr) \land \texttt{st->count} = \length(\tr) \land \texttt{st->sum} = \sumfunc(\tr)}$%
    int i = avg_get_input();
    %$\proofstate{\History(\Concat{\tr}{i}) \land \texttt{st->count} = \length(\tr) \land \texttt{st->sum} = \sumfunc(\tr)}$%    
    st->count += 1;
    %$\proofstate{\History(\Concat{\tr}{i}) \land \texttt{st->count} = \length(\Concat{tr}{i}) \land \texttt{st->sum} = \sumfunc(\tr)}$%        
    st->sum += i;
    %$\proofstate{\History(\Concat{\tr}{i}) \land \texttt{st->count} = \length(\Concat{tr}{i}) \land \texttt{st->sum} = \sumfunc(\Concat{\tr}{i})}$%
    %$\proofstate{\inv{\Concat{tr}{i},\mathtt{st}}}$%    
    avg_unlock(st); 
  }
}

void avg_declass_thread() {
  struct avg_state * st = avg_lock();
  %$\proofstate{\inv{tr,\mathtt{st}}}$%
  %$\proofstate{\History(\tr) \land \texttt{st->count} = \length(\tr) \land \texttt{st->sum} = \sumfunc(\tr)}$%
  if (st->count >= 6) {
    %$\proofstate{\History(\tr) \land \texttt{st->count} = \length(\tr) \land \texttt{st->sum} = \sumfunc(\tr) \land \sumfunc(\tr) \geq 6}$%    
    int avg = st->sum / st->count;
    %$\proofstate{\History(\tr) \land \underline{\texttt{st->count} = \length(\tr) \land \texttt{st->sum} = \sumfunc(\tr) \land \sumfunc(\tr) \geq 6 \land \texttt{avg} = \sumfunc(\tr) / \length(\tr)}}$%
    _(assume avg :: low) %$\hfill \prooftext{\triangleright \{(P,\tr,\texttt{avg} :: \texttt{\bfseries low})\}}$%
    %$\proofstate{\History(\tr) \land \texttt{st->count} = \length(\tr) \land \texttt{st->sum} = \sumfunc(\tr) \land \sumfunc(\tr) \geq 6 \land \texttt{avg} = \sumfunc(\tr) / \length(\tr) \land \texttt{avg} :: \texttt{\bfseries low}}$%
    print_average(avg);
    %$\proofstate{\History(\tr) \land \texttt{st->count} = \length(\tr) \land \texttt{st->sum} = \sumfunc(\tr) \land \sumfunc(\tr) \geq 6 \land \texttt{avg} = \sumfunc(\tr) / \length(\tr) \land \texttt{avg} :: \texttt{\bfseries low}}$%    
  }
  %$\proofstate{\History(\tr) \land \texttt{st->count} = \length(\tr) \land \texttt{st->sum} = \sumfunc(\tr)}$%
  %$\proofstate{\inv{tr,\mathtt{st}}}$%
  avg_unlock(st);
}
\end{lstlisting}
\end{added}

\caption{\changed{}{Proof of the example in \cref{fig:motivation}. The use of the
  assume statement induces an audit triple
  $(P,\tr,\texttt{avg :: \texttt{\bfseries low}})$, where $P$ is the
  underlined
  path condition at the point of the assume statement, namely
  $\texttt{st->count} = \length(\tr) \land \texttt{st->sum} = \sumfunc(\tr) \land \sumfunc(\tr) \geq 6 \land \texttt{avg} = \sumfunc(\tr) / \length(\tr)$.
  The declassification
  policy $\policy{\Dpre}{\Dpost}$ is honored (\cref{defn:audit}) if
  $P$ implies~$\Dpre$
  and 
  and $P \star \Dpost$ implies the assumption $\texttt{avg} :: \texttt{\bfseries low}$.
  The policy for this example recall is
  $\Dpre(\tr)\  \defeq\ \length(\tr) \ge 6$, and
  $\Dpost(\tr)\  \defeq\  \sumfunc(\tr)/\length(\tr) \haslabel \low$.
  $P$ clearly implies $\Dpre$; moreover, so does
  $P \star \Dpost$ imply $\texttt{avg} :: \texttt{\bfseries low}$.
  Thus the example is secure against the declassification policy,
  by \cref{thm:audit}.}\label{fig:motivation-proof}}
\end{figure*}

\end{added}

\section{Proofs of the Main Theorems}\label{sec:proofs}

\begin{added}
The proofs of the main theorems \cref{thm:guarantee}
and \cref{thm:audit}
are expressed with respect to an inductive generalization
that captures all necessary properties of two executions running in lockstep.
The respective soundness proofs will therefore rely on an intermediate result
that precisely characterizes how a major run is related to any minor run.


\begin{definition}[Aligned actions and schedules]
    \label{defn:aligned}
Two actions~$a$ and~$a'$ are \emph{aligned} wrt. a security level~$\ell$,
written $a \cong_\ell a'$, if one of the listed cases applies.
\changed{Two schedules of the same length are aligned, written $\sigma \cong_\ell \sigma'$,
if their actions are point-wise aligned.}{Two schedules are aligned, written $\sigma \cong_\ell \sigma'$,
if they have the same length and their actions are point-wise aligned.}
{\normalfont
\begin{align*}
\tau &\cong_\ell \tau
    \rlap{\qquad $\Fst \cong_\ell \Fst$}
    &
\Snd &\cong_\ell \Snd
    &
\Load{p} &\cong_\ell \Load{p} \\
\Assumption{s~\rho} &\cong_\ell \Assumption{s'~\rho}
    &
\Trace{e} &\cong_\ell \Trace{e'}
    &
\Store{p} &\cong_\ell \Store{p} \\[4pt]
\Out{\ell'}{v} &\cong_\ell
    \rlap{$\Out{\ell'}{v'}$ \quad if $\ell \sqsubseteq \ell' \implies v = v'$}
\end{align*}}
\end{definition}
Aligned schedules capture
that the type of events and formula~$\rho$ for assumptions matches per step,
and that an attacker cannot learn information form memory access
and from outputs (equality of pointers $p$ as well as values~$v$, $v'$).
Note that the condition is strictly stronger than observably equivalent schedules $\sigma \approx_\ell \sigma'$ (\cref{defn:equivalent}),
specifically, it enforces that the event type is always the same
even for unobservable events
and that assumption steps are paired with the \emph{same} assumed formula~$\rho$.

Soundness is characterized
with the help of an inductive predicate $\secure{\ell}{n}{P_1,\tr_1,c,Q,\tr,A}$
which states that the program is safe to execute, correct, and secure for~$n$ steps
similar to \citep[Def.~3]{Ernst_Murray_19} for \SecCSL and its
non-relational counterpart from \citep{vafeiadismfps11} for standard CSL.
In comparison to \citep{Ernst_Murray_19} it \emph{additionally}
tracks alignment between possible schedules of the execution of~$c$,
injects assumption steps into intermediate path conditions,
tracks the history trace of events, and collects audit triples in~$A$.
As such, it encodes all consequences of extended judgements
$\audit{\ell}{P_1 \star \History(\tr_1)}{c}{Q\star \History(tr)}{A}$
to prove \cref{thm:guarantee,thm:audit}
but while this judgement is defined compositionally over the \emph{structure of the program command~$c$},
predicate $\secure{\ell}{n}{P_1,\tr_1,c,Q,\tr,A}$ unwinds individual \emph{execution steps} linearly.
This is the key gap that is bridged in the soundness proof.

\begin{definition}[Secure Executions]
Predicate $\mathsf{secure}$ is defined recursively over the counter~$n$ of remaining steps to assert secure:
\begin{itemize}
\item $\secure{\ell}{0}{P_1,\tr_1,c,Q,\tr,A}$ holds always, i.e., a program is secure for zero steps.
\item $\secure{\ell}{n+1}{P_1,\tr_1,c,Q,\tr,A}$ holds, if for all possible pairs of first steps $\Run{L_1}{c_1}{s_1}{h_1} \xrightarrow{\sigma_1} k_2 $ and $\Run{L}{c_1}{s'_1}{h'_1} \xrightarrow{\sigma'_1} k'_2$
starting from states
$(s_1,h_1),(s'_1,h'_1) \models P_1 \star \History(tr_1) \star \invs{L_1}$
we have
    \begin{enumerate}
      \item  $\sigma_1 \cong_\ell \sigma'_1$ are aligned according to \cref{defn:aligned}, and
      \item the two configurations $k_2$ and $k'_2$ are \emph{matched}, in the sense that either both are $\textbf{stop}$ped with the same lock-set $L_2$ or both are $\textbf{run}$ning with identical commands~$c_2$ and lock-set~$L_2$, and
      \item if the execution step was an assumption~$\rho$, $A$ must contain a corresponding audit triple~$(P'_1,\tr_1,\rho)$ for current trace~$\tr_1$ and some assertion~$P'_1$ that follows from the current path condition~$P_1$, i.e., $P_1 \implies P'_1$ (weakening is allowed and necessary to validate the consequence rule), and
      \item if  $k_2 = \Stop{L_2}{s_2}{h_2}$
            and $k'_2 = \Stop{L_2}{s'_2}{h'_2}$
            then
            either the step was a violated assumption~$\rho$ with
            $s_1,s'_1 \not\models \rho$
            or the postcondition holds
            $(s_1,h_1),(s'_1,h'_1) \models Q \star \History(\tr_2) \star \invs{L_2}$
            for some $\tr_2$, and
      \item if  $k_2 = \Run{c_2}{L_2}{s_2}{h_2}$
            and $k'_2 = \Run{c_2}{L_2}{s'_2}{h'_2}$
            then
            either the step was a violated assumption~$\rho$ with
            $s_1,s'_1 \not\models \rho$
            or some intermediate assertion $P_2$ holds
            $(s_1,h_1),(s'_1,h'_1) \models P_2 \star \History(\tr_2) \star \invs{L_2}$
            for some $\tr_2$
            as well as recursively,
            the program is secure for the remaining~$n$ steps from
            that point on, i.e., $\secure{\ell}{n}{P_2,\tr_2,c_2,Q,\tr,A}$.
      \end{enumerate}
\end{itemize}
\end{definition}
So far, we have tacitly suppressed the semantic model underlying the
abstract predicate $\History(\tr)$.
It can be explained either by introducing a constant ghost location $\mathsf{tr}$
in heaps, so that $(s,h),(s',h') \models \History(\tr)$ iff
$h = [\mathsf{tr} \mapsto \sem{\tr}{s}]$
and
$h' = [\mathsf{tr} \mapsto \sem{\tr}{s'}]$
(recall that $\tr$ is an expression and cf. \citep{banerjee2008expressive,Schoepe_MS_20}),
or alternatively we can interpret $\tr$ in terms of yet another type of actions in the schedule
(which we have done in our Isabelle/HOL proofs).

\begin{lemma}
A valid proof using the rules
$\audit{\ell}{P_1 \star \History(\tr_1)}{c}{Q\star \History(tr)}{A}$,
implies that the program is secure for any number of steps,
i.e.,
$\forall n.\ \secure{\ell}{n}{P_1,\tr_1,c,Q,\tr,A}$.
\end{lemma}
\begin{proof}
By induction on the derivation of 
$\audit{\ell}{P_1 \star \History(\tr_1)}{c}{Q\star \History(tr)}{A}$.
Structural rules (frame, conseq) and compound statements (if, while, sequential and parallel composition)
need an inner induction on the number of steps~$n$.
In our mechanized development, each case is formulated as a separate lemma.
\end{proof}

\begin{lemma}[Secure, Lock-step runs]
    \label{lem:secure}
Assume for all~$n$, that $\secure{\ell}{n}{P_1,\tr_1,c,Q,\tr,A}$.
For a major run $\Run{L}{c}{s}{h} \xrightarrow{\sigma} k$
and a minor run $\Run{L}{c}{s'}{h'} \xrightarrow{\sigma'} k'$
with the \emph{same} program~$c$, lockset~$L$ and $|\sigma| = |\sigma'|$
such that $(s,h),(s',h') \models_\ell P_1 \star \History(\tr_1) \invs{L}$,
we have
\begin{itemize}
      \item the two configurations $k$ and $k'$ are \emph{matched}, in the sense that either both are $\textbf{stop}$ped with the same lock-set $L_2$ or both are $\textbf{run}$ning with identical commands~$c_2$ and lock-set~$L_2$,
\end{itemize}
      and one of the following is true:
\begin{itemize}
\item an assumption has failed at some step $m < |\sigma|$, i.e.,
predicate
      $\mathsf{assumption\text{-}failed}_\ell(m, \sigma, \sigma')$
      holds and the two prefix runs
      given as
      $\Run{L}{c}{s}{h} \xrightarrow{\sigma\cut{m+1}} k_{m+1}$
      and $\Run{L}{c}{s'}{h'} \xrightarrow{\sigma'\cut{m+1}} k'_{m+1}$
      are characterized as follows
      \begin{enumerate}
      \item the intermediate configurations $k_{m+1}$ and $k'_{m+1}$ are matched (cf. above), and
      \item the schedules are aligned up to and including that step, and
            $\sigma\cut{m+1} \cong_\ell \sigma'\cut{m+1}$, and
      \item there is an intermediate assertion~$P_k$ and trace expression~$\tr_k$
            so that $(P_k,\tr_k,\rho) \in A$,
            where $\sigma(k) = \sigma'(k) = \Assumption{\rho}$ is the failed assertion,
            and $P_k \star \History(\tr_k) \star \invs{L_k}$
            holds in the states of $k_{m+1}$ and $k'_{m+1}$ for the corresponding lock-set $L_k$.
      \end{enumerate}
\item no assumption has failed and the entire runs are aligned
      $\sigma \cong_\ell \sigma'$,
      where $k$ and $k'$ either both stopped and validate postcondition~$Q \star \History(\tr') \star \invs{L'}$
      for the some trace expression~$\tr'$ and respective lock-set~$L'$ of~$k$ and~$k'$,
      or they are both running with the same residual program~$c'$
      and similarly validate some intermediate assertion~$P'$ and trace expression~$\tr'$
      from which $c'$ is again secure for any number of steps,
      $\forall n.\ \secure{\ell}{n}{P',\tr',c',Q,\tr,A}$.
\end{itemize}
\end{lemma}
\begin{proof}[Sketch]
This lemma is proved by induction on the (locked) steps of the two executions,
unfolding the inductive security property alongside.
\end{proof}

\begin{lemma}
    \label{lem:audit}
\changed{Given}{Consider} a verified program
$\audit{\ell}{P_1 \star \History(\tr_1)}{c}{Q \star \History(\tr')}{A}$,
i.e., $\forall n.\ \secure{\ell}{n}{P_1,\tr_1,c,Q,\tr,A}$,
a policy~$D$ that has been formally audited (\cref{defn:audit})
and
a pair of a major run
$\Steps{\Run{L_1}{c}{s_1}{h_1}}{\sigma_1}{k_1}$
and minor run
$\Steps{\Run{L_1}{c}{s_1'}{h_1'}}{\sigma'_1}{k'_1}$
from the precondition $(s_1,h_1),(s'_1,h'_1) \models P_1 \star \invs{L_1}$.
\changed{Each}{Then each} assumption failure at some $m < |\sigma| = |\sigma'|$
with
$\sigma(m) = \Assumption{s_m~\rho}$,
$\sigma'(m) = \Assumption{s'_m~\rho}$
is paired with an entry with $(P_m,\tr_m,\rho) \in A$ and
intermediate states
$L_m,s_m,h_m,s'_m,h'_m$
with
$(s_m,h_m),(s'_m,h'_m) \models P_m \star \invs{L_m}$,
such that the specified traces match the schedule:
$\sem{\tr_m}{s_m} = \Concat{\sem{\tr_1}{s_1}}{\trace{\sigma\cut{m}}}$
and
$\sem{\tr_m}{s'_m} = \Concat{\sem{\tr_1}{s'_1}}{\trace{\sigma'\cut{m}}}$.
\end{lemma}
\begin{proof}[Sketch]
This lemma is proved by induction on the (locked) steps of the two executions,
unfolding the inductive security property alongside.
\end{proof}


\changed{
We have mechanized this proof in Isabelle/HOL
with the help of an inductive predicate $\secure{\ell}{n}{P,c,Q}$
which states that the program is secure for~$n$ steps,
based on the existing formalization of \citet{Ernst_Murray_19}.
The first bullet-point of \cref{lem:secure}
and the conclusion of alignment in the second one is added,
and we inject assumptions made by steps
in the definition of $\secure{\ell}{n}{P,c,Q}$.
This lemma therefore reflects consequences on the newly introduced notions,
i.e., much more detailed information about the respective executions
that is now present in the schedules.
}{}

%
%
%
%

\end{added}

\renewcommand{\thetheorem}{\ref{thm:guarantee}}
\begin{theorem}[Policy-agnostic security guarantee]%
If $\prove{\ell}{P}{c}{Q}$ then for a major run
$\Steps{\Run{L}{c}{s}{h}}{\sigma_1}{k_1}$
the knowledge gain from one additional step $\Step{k_1}{\sigma_2}{k_2}$,
expressed as the difference in uncertainty,
is bounded by the release condition:
\begin{align*}
& \uncertainty{\ell}{P}{\sigma_1}{c}{L}{s}{h} \setminus
  \uncertainty{\ell}{P}{\Concat{\sigma_1}{\sigma_2}}{c}{L}{s}{h} \\
{} \subseteq {} & \assumedrelease{\ell}{P}{\sigma_1}{c}{L}{s}{h}
\end{align*}
\end{theorem}
\begin{proof}[Sketch] 
Unfolding the definitions, we obtain a minor run in the set difference
$\Steps{\Run{L}{c}{s'}{h'}}{\sigma'_1}{k'_1}$
that is still uncertain, i.e., with $\sigma_1 \simeq_\ell \sigma'_1$,
but none of its extensions are.
Considering the two cases from \cref{lem:secure}, noting that
$k_1 = \Run{L}{c_1}{s_1}{h_1}$ must be running
and $\prove{\ell}{P_1}{c_1}{Q}$ for some $P_1$.
\begin{itemize}
\item If this pair of runs already contains a failed assumption,
      then it witnesses the release condition, even if the attacker has not been able to observe any consequence of that fact yet.
\item Otherwise, since $k_1$ produces another step
      we have a matching $\Step{k'_1}{\sigma_2}{k'_2}$
      (both $k_1$ and $k'_1$ are running and the small-step semantics is left-total),
      and
      this extension leaks information, i.e., $\sigma_2 \not\simeq_\ell \sigma'_2$.
      Applying \cref{lem:secure} again from $P_1$ for just that step produces a contradiction:
      because we have the stronger condition $\sigma_2 \cong_\ell \sigma'_2$
      from assumption steps (which are invisible) \emph{and} for regular steps
      (which are proven secure).
\end{itemize}
\end{proof}

\renewcommand{\thetheorem}{\ref{thm:audit}}
\begin{theorem}[Policy-specific security guarantee]
For a verified program
$\audit{\ell}{P \star \History(\List{})}{c}{Q \star \History(\tr')}{A}$
and a policy~$D$ formally audited according to \cref{defn:audit}
for each major run
$\Steps{\Run{L}{c}{s}{h}}{\sigma_1}{k_1}$
with final step
$\Step{k_1}{\sigma_2}{k_2}$:
\begin{align*}
                & \assumedrelease{\ell}{P}{\sigma}{c}{L}{s}{h} \\
{} \subseteq {} & \policyrelease{\ell}{\D}{P}{\sigma}{c}{L}{s}{h}
\end{align*}
\end{theorem}
\begin{proof}[Sketch]
Fix a minor run with schedule~$\sigma'$ from $\mathsf{assumed\text{-}release}$
that has an assumption failure at step
with respect to the major run.
By \cref{lem:secure} this occurs at some point~$n$
up to which ${\sigma_1}\cut{n} \cong {\sigma'_1}\cut{n}$
which is critical for some side-conditions.
By \cref{lem:audit} for $\tr_1 = \List{}$, $P_1 = P$,
and the runs up to step~$n$ we obtain an corresponding audit triple
$(P_2,\tr',\rho)$
for which \cref{defn:audit}
guarantees $\Dpre$ is implied by that $P_2$,
but the failed assumption~$\rho$ falsifies $\Dpost$
(contraposition of the second implication of the audit),
and therefore \cref{defn:policy-excludes} is satisfied.
\end{proof}

\begin{added}
\section{Program Semantics}\label{app:sem}

The single-step operational semantics is defined in
\cref{fig:sem-cmd}; multiple steps of execution $\Steps{k}{\sigma}{k'}$ is
defined inductively below. The
assertion semantics are defined in \cref{fig:sem-assert}.

\begin{figure*}
  \begin{added}
\begin{mathpar}
\infer{s' = s(x := \sem{e}{s})}
      {\Step{\Run{x := e}{L}{s}{h}}{\List{\Tick}}{\Stop{L}{s'}{h}}}
      
\infer{\sem{e}{s} \notin \dom{h}}
      {\Step{\Run{\LOAD{x}{e}}{L}{s}{h}}{\List{\Load{\sem{e}{s}}}}{\Abort}}
      
\infer{\sem{e}{s} \in \dom{h} \and s' = s(x := h(\sem{e}{s}))}
      {\Step{\Run{\LOAD{x}{e}}{L}{s}{h}}{\List{\Load{\sem{e}{s}}}}{\Stop{L}{s'}{h}}}

\infer{\sem{e_1}{s} \notin \dom{h}}
      {\Step{\Run{\STORE{e_1}{e_2}}{L}{s}{h}}{\List{\Store{\sem{e_1}{s}}}}{\Abort}}
      
\infer{\sem{e_1}{s} \in \dom{h} \and h' = h(\sem{e_1}{s} \mapsto \sem{e_2}{s})}
      {\Step{\Run{\STORE{e_1}{e_2}}{L}{s}{h}}{\List{\Store{\sem{e_1}{s}}}}{\Stop{L}{s}{h'}}}
      
\infer{l \in L \and L' = L \setminus \{l\}}
      {\Step{\Run{\LOCK{l}}{L}{s}{h}}{\List{\Tick}}{\Stop{L'}{s}{h}}}
      
\infer{l \notin L \and L' = L \union \{l\}}
      {\Step{\Run{\UNLOCK{l}}{L}{s}{h}}{\List{\Tick}}{\Stop{L'}{s}{h}}}
      
\infer{\Step{\Run{c_1}{L}{s}{h}}{\sigma}{\Abort}}
      {\Step{\Run{c_1; c_2}{L}{s}{h}}{\sigma}{\Abort}}
      
\infer{\Step{\Run{c_1}{L}{s}{h}}{\sigma}{\Stop{L'}{s'}{h'}}}
      {\Step{\Run{c_1; c_2}{L}{s}{h}}{\sigma}{\Run{c_2}{L'}{s'}{h'}}}
      
\infer{\Step{\Run{c_1}{L}{s}{h}}{\sigma}{\Run{c'_1}{L'}{s'}{h'}}}
      {\Step{\Run{c_1;c_2}{L}{s}{h}}{\sigma}{\Run{c'_1;c_2}{L'}{s'}{h'}}}
          
\infer{\Step{\Run{c_1}{L}{s}{h}}{\sigma}{\Abort}}
      {\Step{\Run{c_1 \parallel c_2}{L}{s}{h}}{\Concat{\List{\Fst}}{\sigma}}{\Abort}}
      
\infer{\Step{\Run{c_1}{L}{s}{h}}{\sigma}{\Stop{L'}{s'}{h'}}}
      {\Step{\Run{c_1 \parallel c_2}{L}{s}{h}}{\Concat{\List{\Fst}}{\sigma}}{\Run{c_2}{L'}{s'}{h'}}}
      
\infer{\Step{\Run{c_1}{L}{s}{h}}{\sigma}{\Run{c'_1}{L'}{s'}{h'}}}
      {\Step{\Run{c_1 \parallel c_2}{L}{s}{h}}{\Concat{\List{\Fst}}{\sigma}}{\Run{c'_1 \parallel c_2}{L'}{s'}{h'}}}

\infer{\text{if } s \vDash e \text{ then } c' = c_1 \text{ else } c' = c_2}
      {\Step{\Run{\ITE{e}{c_1}{c_2}}{L}{s}{h}}{\List{\Tick}}{\Run{c'}{L}{s}{h}}}
      
\infer{s \not\vDash e}
      {\Step{\Run{\WHILE{e}{c}}{L}{s}{h}}{\List{\Tick}}{\Stop{L}{s}{h}}}
      
\infer{s \vDash e}
      {\Step{\Run{\WHILE{e}{c}}{L}{s}{h}}{\List{\Tick}}{\Run{c; \WHILE{e}{c}}{L}{s}{h}}}
      
\infer{\ }{\Step{\Run{\SKIP}{L}{s}{h}}{\List{\Tick}}{\Stop{L}{s}{h}}}

\infer{\ }{\Step{\Run{\ASSUME{\rho}}{L}{s}{h}}{\List{\Assumption{\rho}}}{\Stop{L}{s}{h}}}

\infer{\ }{\Step{\Run{\OUTPUT{\el}{\ev}}{L}{s}{h}}{\List{\Out{\sem{\el}{s}}{\sem{\ev}{s}}}}{\Stop{L}{s}{h}}}

\infer{\ }{\Run{\TRACE{e}}{L}{s}{h}
  \xrightarrow{~\List{\Trace{\sem{e}{s}}}~}
    \Stop{L}{s}{h}}
\end{mathpar}
\end{added}
  \caption{\changed{}{Small-step operational semantics.\label{fig:sem-cmd} Symmetric parallel rules in which $c_2$ is scheduled
producing the event $\Snd$ have been omitted in the interests of brevity. We write $s \vDash e$ when evaluating expression~$e$ in state~$s$ and casting the resulting value to a boolean yields the value $\trueconst$; we write $s \not\vDash e$ otherwise.}
}
\end{figure*}

\begin{mathpar}
  \infer{\ }{\Steps{k}{\Nil}{k}}
  
  \infer{\Step{k}{\sigma_1}{k_1} \and \Steps{k_1}{\sigma_2}{k_2}}{\Steps{k}{(\Concat{\sigma_1}{\sigma_2})}{k_2}}
\end{mathpar}

\end{added}

\fi 

\end{document}